\documentclass[sigconf, nonacm, screen]{acmart}

\usepackage{todonotes}
\usepackage{bm}
\usepackage{algorithm}
\usepackage{algpseudocode}
\usepackage{hyperref}
\usepackage{csquotes}
\usepackage[capitalize]{cleveref}
\crefname{section}{Sec.}{Secs.}
\usepackage{subcaption}
\usepackage{xspace}
\usepackage{afterpage}
\usepackage{braket}
\usepackage{todonotes}
\AtBeginDocument{%
  }

\acmConference[FPGA '24]{32nd ACM/SIGDA International Symposium on Field-Programmable Gate Arrays}{March 03--05,
2024}{Monterey, CA, USA}
\theoremstyle{thmstyleone}

\newtheorem{lemma}{Lemma}

\Crefname{definition}{Def.}{Defs.}
\Crefname{proposition}{Prop.}{Props.}

\theoremstyle{thmstyletwo}%

\theoremstyle{thmstylethree}%
\newtheorem{definition}{Definition}%
\renewcommand{\vec}[1]{\bm{#1}}
\newcommand{\mat}[1]{\bm{#1}}
\renewcommand{\set}[1]{\mathcal{#1}}
\newcommand{\row}{A}
\newcommand{\column}{B}
\newcommand{\cycle}{C}
\newcommand*{\defeq}{\mathrel{\vcenter{\baselineskip0.5ex\lineskiplimit0pt\hbox{\scriptsize.}\hbox{\scriptsize.}}}%
	=}

\DeclareMathOperator*{\argmax}{arg\,max}
\DeclareMathOperator*{\argmin}{arg\,min}
\DeclareMathOperator{\vect}{vec}
\DeclareMathOperator*{\diag}{diag}
\DeclareMathOperator{\tr}{tr}
\newcommand{\QUBO}{\textsc{Qubo}\@\xspace}
\newcommand{\FPGA}{\textsc{FPGA}\@\xspace}
\newcommand{\QAP}{\textsc{QAP}\@\xspace}
\newcommand{\place}{\textsc{Placement}\@\xspace}
\newcommand{\CX}{\textsc{Cyclic}\textsc{Expansion}\@\xspace}
\newcommand{\cost}[1]{c\left(#1\right)}
\begin{document}

\title{FPGA-Placement via Quantum Annealing}

\author{Thore Gerlach}
\orcid{...}
\affiliation{%
  \institution{Fraunhofer IAIS}
  \city{Sankt Augustin}
  \country{Germany}}
\email{thore.gerlach@iais.fraunhofer.de}

\author{Stefan Knipp}
\orcid{...}
\affiliation{%
  \institution{Thales SIX}
  \city{Ditzingen}
  \country{Germany}}
\email{Stefan.KNIPP@thalesgroup.com}
%0009-0003-1910-3130

\author{David Biesner}
\orcid{...}
\affiliation{%
 \institution{Fraunhofer IAIS}
 \city{Sankt Augustin}
 \country{Germany}}
\email{david.biesner@iais.fraunhofer.de}
 
\author{Stelios Emmanouilidis}
\orcid{...}
\affiliation{%
	\institution{Fraunhofer IAIS}
	\city{Sankt Augustin}
	\country{Germany}}
\email{stelios.emmanouilidis@iais.fraunhofer.de}
%0009-0004-0547-5320

% \author{Lars Doepper}
% \affiliation{%
% 	\institution{Fraunhofer IAIS}
% 	\city{Sankt Augustin}
% 	\country{Germany}}
% \email{lars.doepper@iais.fraunhofer.de}

\author{Klaus Hauber}
\orcid{...}
\affiliation{%
	\institution{Thales SIX}
	\city{Ditzingen}
	\country{Germany}}
\email{Klaus.HAUBER@thalesgroup.com}
% 0009-0009-8103-7983

\author{Nico Piatkowski}
\orcid{...}
\affiliation{%
	\institution{Fraunhofer IAIS}
	\city{Sankt Augustin}
	\country{Germany}}
\email{nico.piatkowski@iais.fraunhofer.de}
% 0000-0002-6334-8042

\renewcommand{\shortauthors}{Gerlach et al.}

\begin{abstract}
	Field-Programmable Gate Arrays (FPGAs) have asserted themselves as vital assets in contemporary computing by offering adaptable, reconfigurable hardware platforms. 
	FPGA-based accelerators incubate opportunities for breakthroughs in areas, such as real-time data processing, machine learning or cryptography---to mention just a few. 
	The procedure of placement---determining the optimal spatial arrangement of functional blocks on an FPGA to minimize communication delays and enhance performance---is an NP-hard problem, notably requiring sophisticated algorithms for proficient solutions.
	Clearly, improving the placement leads to a decreased resource utilization during the implementation phase.
	Adiabatic quantum computing (AQC), with its capability to traverse expansive solution spaces, has potential for addressing such combinatorial problems. 
	% unfolds novel horizons for addressing this complexity.
	In this paper, we re-formulate the placement problem as a series of so called quadratic unconstrained binary optimization (\QUBO) problems which are subsequently solved via AQC.
	Our novel formulation facilitates a straight-forward integration of design constraints.
	Moreover, the size of the sub-problems can be conveniently adapted to the available hardware capabilities.
	Beside the sole proposal of a novel method, we ask whether contemporary quantum hardware is resilient enough to find placements for real-world-sized FPGAs. 
%	Quantum algorithms, such as the Quantum Approximate Optimization Algorithm (QAOA), can be harnessed to navigate through the enormous combinatorial space of possible placements, discerning global minima that classical solvers might miss or take inordinate time to find. 
%	This interplay between quantum computing and FPGA placement not only stands to elevate the efficiency and performance of resultant FPGA configurations but also
	A numerical evaluation on a D-Wave Advantage 5.4 quantum annealer suggests that the answer is in the affirmative. 
\end{abstract}

\begin{CCSXML}
	<ccs2012>
	<concept>
	<concept_id>10010583.10010600.10010628</concept_id>
	<concept_desc>Hardware~Reconfigurable logic and FPGAs</concept_desc>
	<concept_significance>500</concept_significance>
	</concept>
	<concept>
	<concept_id>10010583.10010786.10010813.10011726</concept_id>
	<concept_desc>Hardware~Quantum computation</concept_desc>
	<concept_significance>500</concept_significance>
	</concept>
	<concept>
	<concept_id>10003752.10003809.10003716.10011136</concept_id>
	<concept_desc>Theory of computation~Discrete optimization</concept_desc>
	<concept_significance>300</concept_significance>
	</concept>
	<concept>
	<concept_id>10002950.10003624.10003625.10003627</concept_id>
	<concept_desc>Mathematics of computing~Permutations and combinations</concept_desc>
	<concept_significance>300</concept_significance>
	</concept>
	</ccs2012>
\end{CCSXML}

\ccsdesc[500]{Hardware~Reconfigurable logic and FPGAs}
\ccsdesc[500]{Hardware~Quantum computation}
\ccsdesc[300]{Theory of computation~Discrete optimization}
\ccsdesc[300]{Mathematics of computing~Permutations and combinations}

\keywords{Quantum Computing, QUBO, FPGA, Placement, QAP, Permutations}


\maketitle

\begin{figure*}[t]
	\centering
	\includegraphics[width=0.9\textwidth]{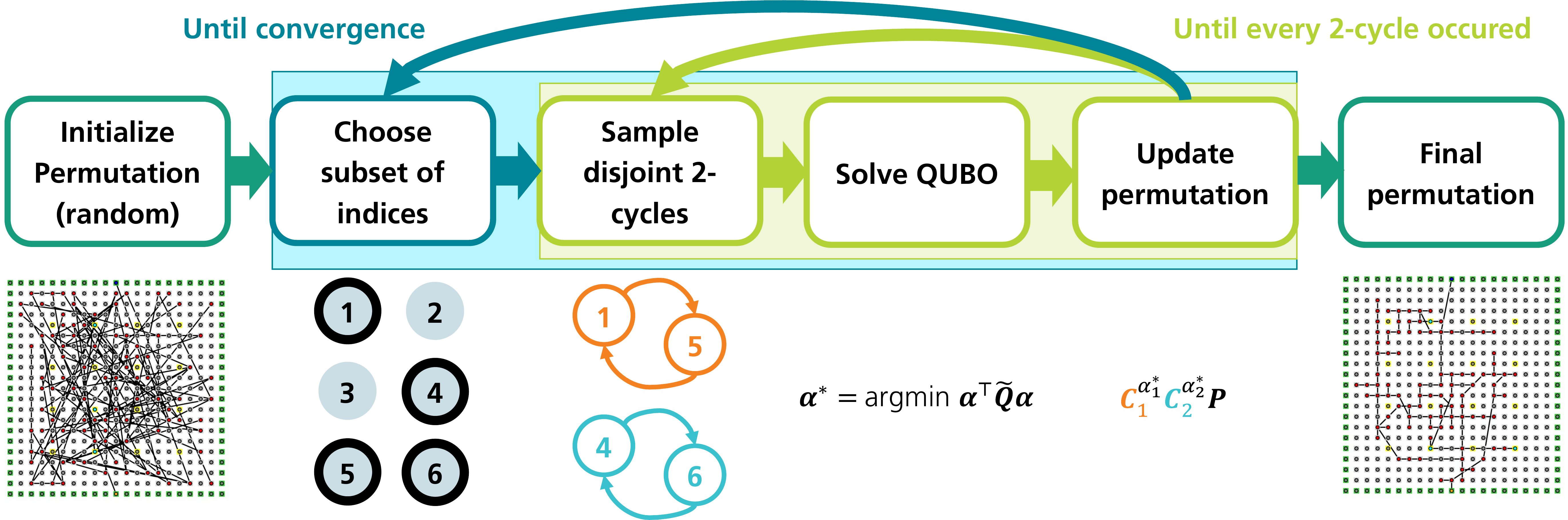}
	\caption{Flow chart of the \CX (\cref{alg:expansion}): Given an initial permutation (random), we choose a sub-problem and iteratively sample random disjoint cycles, which are used to formulate a \QUBO problem.
	This \QUBO formulation is solved with quantum annealing, giving us a binary vector $\vec{\alpha}^*$, indicating which cycle should be applied to the current permutation. If every cycle occurred, we choose a new sub-problem and repeat this procedure until convergence.}
	\label{fig:algo}
\end{figure*}

\section{Introduction}
%
%FPGAs allow us to implement sophisticated and demanding algorithms on programmable hardware. 
%Performance and costwise, FPGAs are between CPUs on the one side, which are low performance, low cost, and ASICs on the other side, which can reach high performance but come at very high cost, depending on quantity.
%
%There's multiple factors why FPGAs still become more and more relevant compared to ASICs nowadays. 
%The initial cost for ASIC designs rises and CPUs are, by architecture, not made for high data rate, parallel computing tasks. This is why mixed designs with a CPU for control tasks with an FPGA for computing tasks get more common.
%
%Today, FPGAs are often used for applications with high data throughput, that are compute intensive or time critical control tasks and for implementing complex communications protocols. A big benefit compared to ASICs is the reconfigurability, i.e. if a product requires new algorithms over its lifetime, they can easily be changed via an update.
%
Logic optimization, placing and routing are fundamental and the most time-consuming steps in the field of chip design for both ASICs and FPGAs \cite{rajavel2011}. 
The number of transistors and logic gates on a single chip is increasing more and more, leading to the mentioned processes consuming more and more time.
This limits the speed of development cycles, which is an issue of productivity but can also be a security issue, 
since faster development cycles for cryptography related algorithms allow for improved security analysis. 

Here, we focus on the placement step, in which we aim to find an ideal placement of functional blocks on the chip.
The advantages of a good placement are two-fold:
On the one hand, minimizing the physical distance between connected elements leads to shorter wire lengths and therefore a higher maximum clock rate.
On the other hand, a good placement can lead to a faster routing process, i.e., the routing algorithm of the connections between elements finding a good solution in fewer iterations and less time.
Since the placement itself is an increasingly time-consuming step, decreasing the runtime of the placement algorithm while maintaining a high solution quality is of great interest.
%
%An accelerated and efficient FPGA development flow is vital for industrial applications.
%First of all, a faster development cycle increases the producitvity during the debugging phase of the design.
%This is crucial especially in the security domain, since a real-world analysis of the security of a cryptographic algorithm is only possible on real hardware and simulation is not sufficient.
%Additionally, a faster development cycle results in the possibility of implementing more designs on FPGAs and therefore more applications benefitting from the very good speed to power consumption ratio of FPGAs.

Taking a closer look at the math behind the placement in the floor-planning case, 
we find that it is equivalent to the quadratic assignment problem (\QAP) \cite{anjos2012, feld2017}.
The goal of the \QAP is to assign each given element to a unique location, minimizing a given cost function.
In chip design, that cost function can be the total wire length between connected units given by a placement on the chip.
Minimizing that function leads to a shorter maximum wire length and with that the possibility for a higher maximum clock rate.
The \QAP is an NP-hard combinatorial optimization problem, and moreover, one of the hardest in this class, since there is no approximation algorithm for producing a sub-optimal solution with guarantees in polynomial time \cite{sahni1976p}. 
Quantum computing (QC), especially adiabatic quantum computing (AQC) \cite{albash2018adiabatic,farhi2000quantum} and its physical implementation of quantum annealing (QA), is very promising in solving such problems better than classical\footnote{The term \enquote{classical} is physics jargon and means \enquote{not quantum}.} algorithms.
This can be done by reformulating the given \QAP to a quadratic unconstrained binary optimization (\QUBO) problem.
%\TG{@Nico Some words to QC and why it is useful for this task}
%\TG{@Nico When to mention evo annealer?}

While real-world quantum devices suffer from a series of technical limitations, there is theoretical evidence that hard combinatorial problems can be solved exactly via AQC. 
We leverage these theoretical insights and describe the first algorithm for placement of functional blocks (e.g., Lookup Tables (LUT), BlockRAMs (BRAM), Digital Signal Processing (DSP)) on an FPGA, based on solving a \QAP via QC.
%We show how our algorithm can solve the \place problem, maintaining or even increasing the solution quality for the investigated designs and therefore providing the basis for an efficient routing step and a fast FPGA development cycle. 
Our contributions can be summarized as follows:
\begin{itemize}
	\item We propose a theoretically sound iterative adiabatic quantum algorithm for solving an unbalanced \QAP with the ability of incorporating constraints on allowed permutations and convenient adaptation of the problem sizes to available hardware capabilities.
	\item We explain how this algorithm can be applied to solve the \FPGA-\place problem.
	\item We provide an experimental evaluation on software and hardware solvers, including real quantum hardware, that proves the viability of our method. 
%	\item A new \QUBO formulation for the unbalanced \QAP.
\end{itemize}
%\TG{More contributions?}
%The remainder of this paper is structured as follows.
%\begin{itemize}
%	\item 
In \cref{sec:related_work}, we give an overview of related work in the field of \FPGA-\place and investigated quantum solutions. \cref{sec:background} contains background information on the mathematical formulation of the \QAP and basics of QC. We describe details of our quantum algorithm for \FPGA-\place in \cref{sec:method}, schematically depicted in \cref{fig:algo}. In \cref{sec:experiments}, we show the validity of our approach with experiments on a fictional \FPGA architecture. We utilize not only software solvers for solving \QUBO problems but real quantum hardware and digital annealing (classical analogue of QA). Lastly, we conclude the results of our paper in \cref{sec:conclusion}.

\section{Related Work} \label{sec:related_work}

The QAP is a traditional combinatorial optimization problem \cite{koopmans1957}, \cite{lawler1963}. 
It is NP-hard and no classical algorithms are known which can approximate a solution with quality guarantees in polynomial time.
The current state-of-the-art methods for solving the floor planning problem/QAP in \FPGA-\place can mainly be divided into three groups: simulated annealing (SA) \cite{betz1997,ludwin2011}, analytical \cite{gort2012,lin2013} and partitioning-based \cite{breuer1977,maidee2005} approaches.
The SA approach can achieve high quality results, especially in terms of subsequent routing time. However, its running time becomes a major drawback when placing a large circuit. 
Contrary to this, partition-based approaches have a very short running time by recursively partitioning a design.
Nevertheless, this might result in bad quality because the problem is solved locally after partitioning.
The analytic approach compromises this quality-speed trade-off, by being very fast and showing similar performances to the SA approach.
Not every functional block contained in the given net list can be placed anywhere on the chip grid\footnote{We stick to the term \enquote{grid} although the placement problem can indeed be lifted to higher dimensions, e.g., $3$-dimensional chips.}, e.g., a LUT cell cannot be placed on an IO location.
The analytic methods need a post-processing for incorporating these constraints.
There also exist other approaches, e.g., ones who are based on machine learning \cite{pui2017,elgamma2020}.
All of the aforementioned methods heavily rely on approximations and often need good initializations.

The idea of addressing hard combinatorial optimization problem, such as \QAP, with quantum computing naturally arises, since quantum computers look promising for overcoming classical limitations.
The most advanced research in this field is given in
\cite{benkner2020,benkner2021,bhatia2023}, where the paradigm of quantum annealing is applied for solving \QAP.
Still, the only work we came across in our literature research which is concerned with solving the \FPGA-\place problem with quantum computing is \cite{guo2009}. 
This paper uses a combination of a quantum genetic algorithm and SA. 

Note that FPGAs are frequently used as control devices in QC hardware \cite{khalid2004,pilch2019,xu2021}.
In this work we apply QC implementations of FPGA designs, which can then be applied to the development of FPGA control units for quantum computers. 
%We discuss this symbiotic relationship between FPGA design and quantum computing in more detail in \cref{sec:conclusion}.

\section{Background} \label{sec:background}

We start off with some theoretical background on (A)QC in \cref{sec:qc} and then move to the \FPGA-\place problem and its related \QAP in \cref{sec:placement}.

\subsection{Notation}

We denote matrices with bold capital letters (e.g. $\mat{A}$) and vectors with bold lowercase letters (e.g. $\vec{a}$).
Furthermore, we use the following standard terms of linear algebra.
The trace of a matrix $\mat{A} \in \mathbb{R}^{n \times n}$ is indicated by $\tr\left(\mat{A}\right)=\sum_{i=1}^{n}A_{i,i}$.
We denote with $\vect\left(\mat{A}\right)\in\mathbb R^{n^2}$ stacking all rows subsequently into a single vector and with $\diag\left(\vec{a}\right)$ the $n\times n$ diagonal matrix with $\vec{a}\in\mathbb R^n$ as its diagonal.
With $\mat{I}_n$ the $n$-dimensional identity matrix is represented and $\vec{1}_n$ denotes the $n$-dimensional vector consisting only of $1$s.
Finally, we represent with $\mathbb P_{n}$ the set of permutation matrices on $n$ elements,
i.e., $\mat{X}\in\{0,1\}^{n\times n}$ with $\mat{X} \vec{1}_n = \vec{1}_n$ and $\mat{X}^{\top} \vec{1}_n = \vec{1}_n$.

\subsection{Practical Quantum Computing} \label{sec:qc}
Let us quickly introduce the basic notion of what can be considered as quantum computation \cite{nielsen2010quantum}. 
Today, practical QC consists of two dominant paradigms: AQC and gate-based QC. In both scenarios, a quantum state $\ket{\psi}$ for a system with $n$ qubits is a $2^n$ dimensional complex 
vector. In the gate-based framework, a quantum computation is defined as a matrix multiplication $\ket{\vec{\psi}_{\text{out}}}=C\ket{\vec{\psi}_{\text{in}}}$, where the $C$ is a $(2^n \times 2^n)$-dimensional unitary matrix (the circuit), given via a 
series of inner and outer products of low-dimensional unitary matrices (the gates). In AQC---the
framework that we consider in the paper at hand---the result of computation is defined to be the eigenvector $\ket{\phi_{\min}}$ that corresponds to the smallest eigenvalue of some $(2^n \times 2^n)$-dimensional Hermitian matrix $\mat{H}$. In 
practical AQC, $\mat{H}$ is further restricted to be a real diagonal matrix whose entries can be 
written as $\mat{H}_{i,i}=\mat{H}(\mat{Q})_{i,i}=\boldsymbol{\vec{x}^{i\top} \mat{Q} \vec{x}^i}$ where $\vec{x}^i=\operatorname{binary}(i)\in\{0,1\}^n$ is some (arbitrary but 
fixed) $n$-bit binary expansion of the unsigned integer $i$. Here, $\mat{Q} \in \mathbb{R}^{n\times n}$ is the so-called \QUBO
matrix. By construction, computing $\ket{\vec{\psi}_{\text{out}}}$ is equivalent to solving $\min_{\vec{x}} \vec{x}^\top \mat{Q} \vec{x}$. Adiabatic quantum algorithms rely on this construction by encoding (sub-)problems as QUBO matrices. 

In both paradigms, the output vector is $2^n$-dimensional and can thus not be read-out efficiently for non-small $n$. Instead, the output of a practical quantum computation is a random integer $i$ between 1 and $2^n$, which is drawn from the probability mass function $\operatorname{Prob}(i)=|\braket{i | \vec{\psi}_{\text{out}}}|^2=|\ket{\vec{\psi}_{\text{out}}}_i|^2$. This sampling step is also known as collapsing the quantum state $\ket{\vec{\psi}_{\text{out}}}$ to a classical binary state $\operatorname{binary}(i)$.

AQC has been applied to numerous hard combinatorial optimization problems \cite{lucas2014}, ranging over satisfiability \cite{kochenberger2005}, routing problems \cite{neukart2017} to machine learning \cite{bauckhage2019}.

\subsection{FPGA-Placement} \label{sec:placement}

%We illustrate the FPGA development procedure on the example of the open-source tool \emph{nextpnr} \cite{shah2019},
%which is used for logic optimization, placement and routing.
%In the first step, nextpnr transforms the logic elements (LUTs, BRAMs, DSPs) from the logic optimization into elements which are actually available on the specific chip.
%The number of elements needed to be placed on the chip, as well as the neccessary conncetions between the elements, is then known exactly.
%After that step, nextpnr generates an initial placement guess, which is used as a seed for a simulated annealing approach.
%Following the placement by simulated annealing, the routing is done by an A* algorithm \cite{astar} with a rip-up and reroute strategy \cite{routing}. 
%The routing time is significantly affected by the quality of the placement, meaning a better placement can increase the routing speed.
%%The quality and time efficiency of an FPGA implementation procedure is largely dominated by the \place and routing steps.
%%The goal of the \place is to determine an optimal spatial arrangement of these blocks on an FPGA to minimize communication delays and enhance performance. 
%%To improve the runtime and the quality of this step, we propose an adiabatic quantum algorithm, which we describe in detail in the following sections.
%%
It is well known that the \place problem can be formulated as an unbalanced \QAP. %, detailed in \cref{sec:qap}.
We now recap this formulation, since our construction in Sec.~\ref{sec:method} relies on it to transform the \place problem into a series of \QUBO problems. %, which can be solved by quantum annealing.

\begin{definition}
	Given are a set of facilities $\set{F}=\{p_1,\dots,p_n\}$ and a set of locations $\set{L}=\{l_1,\dots,l_n\}$, 
	along with a flow function $f:\set{F}\times\set{F}\to\mathbb R$ between facilities and a distance function $d:\set{L}\times\set{L}\to\mathbb R$ between locations.
	We define the flow and distance matrices as
	\begin{equation*}
		\mat{F}\defeq \left(f(p_i,p_j)\right)_{i,j=1}^n,\quad \mat{D}\defeq \left(d(l_i,l_j)\right)_{i,j=1}^n\;.
	\end{equation*}
	We formulate the quadratic assignment problem (\QAP) as
	\begin{equation*}
		\argmin_{\pi}\ \sum_{i,j=1}^n\mat{F}_{i,j}\mat{D}_{\pi(i),\pi(j)}\;,
	\end{equation*}
	where $\pi:[n]\to[n]$ is a permutation on $[n]\defeq \{1,\dots,n\}$.
	An equivalent formulation is given by the corresponding permutation matrices
	\begin{equation*}
		\argmin_{\mat{P}\in\mathbb P_n}\ \tr \left( \mat{F} \mat{P} \mat{D} \mat{P}^{\top} \right)\;.
	\end{equation*}
	For given flow and distance matrices $\mat{F}$ and $\mat{D}$, we define the cost function as
	\begin{equation*}
		\cost{\mat{A},\mat{B}}\defeq\tr \left( \mat{F} \mat{A} \mat{D} \mat{B}^{\top} \right),\quad \cost{\mat{A}}\defeq\cost{\mat{A},\mat{A}}\;.
	\end{equation*}
\end{definition}

In \FPGA-\place the goal is to place $m$ functional blocks into $n$ physical slots on the \FPGA chip grid such that the total wire length is minimized, with $m\le n$.
The $\mat{F}$ indicates how two functional blocks are connected in the given net list and the distance matrix $\mat{D}$ indicates the distances between different locations on the chip.
However, in the \QAP framework we need the two matrices to have the same dimensionality.

\begin{definition}
	Given two sets of indices $\set{I},\set{J}\subset[n]$ with  $\left|\set{I}\right|=k$, $\left|\set{J}\right|=l$, a matrix $\mat{A}\in\mathbb R^{n\times n}$ and a vector $\vec{a}\in\mathbb R^n$.
	We denote the sub-matrix consisting only of the rows indexed by $\set{I}$ and the columns indexed by $\set{J}$ as $\mat{A}_{\set{I},\set{J}}\in\mathbb R^{k\times l}$.
%	For short, we write $\mat{A}_{\set{I}}$ instead of $\mat{A}_{\set{I},\set{I}}$.
	Similarly we denote the sub-vector of $\vec{a}$ consisting only of the entries index by $\set{I}$ as $\vec{a}_{\set{I}}\in\mathbb R^k$.
\end{definition}

Using this definition, we can formalize the \place problem as a \QAP by introducing a new matrix $\mat{F}'\in\mathbb R^{n\times n}$,
which is $0$ everywhere except for its $m\times m$ upper block matrix, i.e., ${\mat{F}'_{[m],[m]} =\mat{F}}$.
Descriptively, we introduce $n-m$ ``dummy'' functional blocks which are not connected to any other unit.
The placement objective can now be written as
\begin{equation}
	\argmin_{\mat{P}\in\mathbb P_n}\ \tr \left( \mat{F}' \mat{P} \mat{D} \mat{P}^{\top} \right)\;.
	\label{eq:qap_matrix}
\end{equation}

Even though it is a common way for obtaining a \QAP formulation, inserting $n-m$ dummy elements leads to a large amount of redundancy and high dimensionality, especially if $m\ll n$.
We can overcome this issue by considering sub-permutations.
\begin{definition}
	With $m,n\in \mathbb N, m\le n$
	we define a map $\pi_{m,n}:[m]\to [n]$ to be a sub-permutation if it is injective.
	$\pi_{m,n}$ can also be described by a binary sub-permutation matrix, whose rows sum to $1$ and whose columns contain at most a single $1$.
	The space of sub-permutation matrices of dimension $m\times n$ can hence be defined as
	\begin{equation*}
		\mathbb P_{m, n}\defeq\{\mat{X}\in\{0,1\}^{m\times n}\mid \mat{X} \vec{1}_n = \vec{1}_m, ~ \mat{X}^{\top} \vec{1}_m \le \vec{1}_n\}\;,
	\end{equation*}
	where the relation $\le$ is to be understood component wise.
\end{definition}
With this definition we can rewrite the objective in \cref{eq:qap_matrix}
\begin{equation}
	\argmin_{\mat{P}\in\mathbb P_{m,n}}\ \tr \left( \mat{F} \mat{P} \mat{D} \mat{P}^{\top} \right)\;.
	\label{eq:qap_sub_permutation}
\end{equation}
For $m<n$, the formulation in \cref{eq:qap_sub_permutation} is also called unbalanced \QAP.

\section{Method} \label{sec:method}

We describe a first quantum-compatible problem formulation in \cref{sec:qubo}. Due to the limited number of qubits available on current quantum annealers, we extend the first formulation to an iterative solving approach in \cref{sec:cx}.
%
%and the move over to the \CX algorithm in \cref{sec:cx}.

\subsection{QUBO Formulation} \label{sec:qubo}

To obtain a \QUBO formulation, we can leverage the optimization over all sub-permutation matrices in \cref{eq:qap_sub_permutation}.
\begin{lemma}
	The \QUBO formulation
	\begin{equation}
		\label{eq:qap_qubo}
		\argmin_{\vec{x} \in \{0,1\}^{mn},\vec{s} \in \{0,1\}^{n}}\ \ \vec{x}^{\top}\mat{Q}\vec{x}+\lambda\left\|\mat{A}\vec{x}-\vec{1}_m\right\|^2+\mu\left\|\mat{B}\vec{x}-\vec{s}\right\|^2\;,
	\end{equation}
	where $\mat{Q}\defeq \mat{F}\otimes \mat{D}$, $\mat{\row}\defeq\mat{I}_m\otimes\vec{1}_n^{\top} $ and $\mat{\column}\defeq\vec{1}_m^{\top}\otimes\mat{I}_n $,
	is equivalent to \cref{eq:qap_sub_permutation} for sufficiently large penalty parameters $\lambda,\mu\in\mathbb R_+$
\end{lemma}
\begin{proof}
	The objective in \eqref{eq:qap_sub_permutation} can be written as
	\begin{align*}
		\argmin_{\mat{X} \in \{0,1\}^{m\times n}}\ \ &\cost{\mat{X}} \\
		\text{subject to} \ \ \ \ &\mat{X} \vec{1}_n = \vec{1}_m, ~ \mat{X}^{\top} \vec{1}_m \le \vec{1}_n\;.
	\end{align*}
	Using $\vec{x}=\vect\left(\mat{X}\right)$ this can be reformulated to
	\begin{align*}
		\argmin_{\mat{x} \in \{0,1\}^{n^2}}\ \ &\vec{x}^{\top}\mat{Q}\vec{x} \\
		\text{subject to} \ \ \ \ &\mat{\row}\vec{x}=\vec{1}_m, ~\mat{\column}\vec{x}=\vec{s},\ \vec{s}\in\{0,1\}^n\;.
	\end{align*}
	With
	\begin{equation*} 
		\mat{\row}\vec{x}=\vec{1}_m, ~\mat{\column}\vec{x}=\vec{s}
		\Leftrightarrow\left\|\mat{\row}\vec{x}-\vec{1}_m\right\|^2=0, ~\left\|\mat{\column}\vec{x}-\vec{s}\right\|^2=0\;,
	\end{equation*}
	we can incorporate the constraints into our objective with using penalty parameters, ending up with \cref{eq:qap_qubo}.
\end{proof}
%$\vec{\lambda},\vec{\mu} \in\mathbb R^n$
%\begin{align}
%	\argmin_{\mat{x} \in \{0,1\}^{n^2}}\ \ \vec{x}^{\top}\mat{Q}\vec{x}+\sum_{i=1}^{n}\lambda_i\left(\left(\mat{\row}\vec{x}\right)_{i,\cdot}-1\right)^2+\mu_i\left(\left(\mat{\column}\vec{x}\right)_{i,\cdot}-1\right)^2 
%	\label{eq:qap_qubo}\;,
%	\\
%	\argmin_{\mat{x} \in \{0,1\}^{n^2}}\ \ \vec{x}^{\top}\mat{Q}\vec{x}+\left(\mat{\row}\vec{x}-\vec{1}\right)^{\top}\left(\lambda\odot\left(\mat{\row}\vec{x}-\vec{1}\right)\right) +\left(\mat{\column}\vec{x}-\vec{1}\right)^{\top}\left(\mu\odot\left(\mat{\column}\vec{x}-\vec{1}\right)\right)
%\end{align}
%similar to the formulation in \cite{benkner2020}.
Even though we have a quantum-compatible problem formulation for the \QAP at hand, we remark that our problem dimension is $mn+n=n(m+1)$. 
That is, for solving an \FPGA-placement problem with $m$ functional blocks and $n$ grid locations, we need $n(m+1)$ qubits, which is beyond capabilities of current (and upcoming) quantum hardware.
Furthermore, choosing the penalty parameters $\lambda,\mu$ maintaining the equivalence between \cref{eq:qap_matrix} and \cref{eq:qap_qubo} while also having preferable conditioning for quantum hardware is tedious and error prone. 
In \cite{benkner2020}, coarse upper bounds are provided for these parameters.
Furthermore, constraints on the permutation space can not easily be integrated into such a formulation.
However, this is of great importance in \FPGA-\place since we have to take into account that the types of every functional block and the corresponding placement location have to match.
For example, it is impossible to place a LUT onto an IO location (cf. \cref{fig:constraints}).

\subsection{Cyclic Expansion} \label{sec:cx}

The above issues can be overcome by not considering a single \QUBO formulation but a series of \QUBO{s}. 
We follow the approach from \cite{benkner2021} and use a variant of the $\alpha$-expansion algorithm \cite{boykov2001}, which we will henceforth denote as \CX.
This reduces the dimensionality of the \QUBO problem and removes the penalty terms in \cref{eq:qap_qubo}. 
Furthermore, constraints on the allowed sub-permutations can be incorporated easily into this algorithm.
The idea is that instead of optimizing over all permutation matrices at once, an iterative optimization over cyclic permutations is carried out which converges towards the original optimization.
For the upcoming sections we assume that $m=n$, the case $m<n$ follows analogously.

Informally, the cyclic expansion algorithm works as follows:
\begin{enumerate}
	\item Initialize permutation matrix $\mat{P}\in\mathbb P_n$,
	\item Choose a set of simpler permutation matrices $\mathbb{\cycle}\subset \mathbb P_n$ with ${|\mathbb{\cycle}| = k < n}$,
	\item Solve a QUBO of size $k$ to decide which permutation in $\mathbb{\cycle}$ to apply to $\mat{P}$,
	\item Update $\mat{P}$ with the chosen permutation,
	\item Repeat steps 2-4 until convergence of cost $c(\mat{P})$.
\end{enumerate}

In what follows, we define all necessary terms and provide a detailed description of \cref{alg:expansion}, also depicted in \cref{fig:algo}.

\begin{definition}
	Define a $2$-cycle $\mat{\cycle}$ as $\mat{\cycle}\in \mathbb P_n$
	\begin{align*}
		\text{such that}\quad&\exists i,j\in[n],i\neq j:\cycle_{i,j}=\cycle_{j,i}=1\\
		\text{and}\quad&\forall l\in[n],l\notin\{i,j\}:\cycle_{l,l}=1\;,
	\end{align*}
	and denote the set of $2$-cycles with $\mathbb P_n^{(2)}$.
\end{definition}

We remark that any permutation matrix $\mat{P}\in\mathbb P_n$ can be written as a product of $2$-cycles, 
\begin{equation}
	\forall \mat{P}\in\mathbb P_n:\exists \{\cycle_1,\dots,\cycle_s\}\subset \mathbb P_n^{(2)}:\mat{P}=\prod_{i=1}^s\mat{\cycle}_i=\mat{\cycle}_1\cdots\mat{\cycle}_s\;.
	\label{eq:power}
\end{equation}
Instead of optimizing over all $2$-cycles the idea of \CX is to iteratively consider fixed subsets of $\mathbb P_n^{(2)}$. 
\begin{definition}
	Let $\mathbb{\cycle}=\{\cycle_1,\dots,\cycle_s\}\subset \mathbb P_n^{(2)}$ be a set of $2$-cycles and let $\alpha\in\{0,1\}^s$. 
	For $\mat{P}\in\mathbb P_n$ we define
	\begin{equation}
		g\left(\mat{P},\vec{\alpha}, \mathbb{\cycle}\right)\defeq \left( \prod_{i=1}^{s} \mat{\cycle}_i ^{\alpha_i}\right) \mat{P}=\mat{\cycle}_1^{\alpha_1}\cdots\mat{\cycle}_s^{\alpha_s}\mat{P}\;,
		\label{eq:label_prod}
	\end{equation}
	with $\mat{\cycle}_i ^0\defeq\mat{I}_n$, $\mat{\cycle}_i ^1\defeq\mat{\mat{\cycle}_i}$.
	In words, the vector $\vec{\alpha}$ indicates which cycle in $\mathbb{\cycle}$ should be applied to $\mat{P}$.
\end{definition}
For a given $\mat{P}\in\mathbb P_n$ the following objective is optimized in each iteration
\begin{equation}
	\argmin_{\vec{\alpha}\in\{0,1\}^s}\ \cost{g\left(\mat{P},\vec{\alpha},\mathbb{\cycle}\right)} \;.
	\label{eq:qap_cyclic}
\end{equation}
However, \cref{eq:qap_cyclic} is not in \QUBO form and can thus not be directly solved on actual quantum hardware. 
To overcome this issue, we only consider disjoint $2$-cycles.
\begin{definition}
	Let $\mat{\cycle},\mat{\cycle}'\in\mathbb P_n^{(2)}$.
	$\mat{\cycle}$ and $\mat{\cycle}'$ are disjoint if 
	\begin{equation*}
		\left(\cycle_{i,i}=0\Rightarrow\cycle'_{i,i}=1\right) \wedge \left(\cycle'_{i,i}=0\Rightarrow\cycle_{i,i}=1\right)\;,
	\end{equation*}
	which leads to commutativity, i.e., $\mat{\cycle}\mat{\cycle}'=\mat{\cycle}'\mat{\cycle}$.
	We call a set $\mathbb{\cycle}\subset \mathbb P_n^{(2)}$ disjoint if all elements are pairwise disjoint.
\end{definition}

Sets of disjoint $2$-cycles have a large expressive power in terms of covering the whole permutation space.
\begin{lemma}
	Given a permutation matrix $\mat{P}$, there exist two sets $\mathbb{\cycle}$, $\mathbb{\cycle}'\subset \mathbb P_n^{(2)}$ of disjoint $2$-cycles such that
	\begin{equation*}
		\mat{P}=\mat{L}\mat{R},\quad\mat{L}\defeq\prod_{\mat{C}\in\mathbb{\cycle}}\mat{C},\quad \mat{R}\defeq\prod_{\mat{C}'\in\mathbb{\cycle}'}\mat{C}'\;.
	\end{equation*}
\end{lemma}	
\begin{proof}
	See \cite{benkner2021}.
\end{proof}
With assuming disjoint $2$-cycles we obtain the following result.
\begin{lemma}
	Assume that $\mathbb{\cycle}=\{\mat{\cycle}_1,\dots,\mat{\cycle}_s\}\subset \mathbb P_n^{(2)}$ is a disjoint set of $2$-cycles and let $\mat{P}\in\mathbb P_n$ be a permutation matrix.
	Then, the following identity holds
	\begin{equation}
		g\left(\mat{P},\vec{\alpha},\mathbb{\cycle}\right)=\mat{P}+\sum_{i=1}^{s}\alpha_i\tilde{\vec{\cycle}}_i,\quad \tilde{\vec{\cycle}}_i\defeq\left(\mat{\cycle}_i-\mat{I}_n\right)\mat{P}\;.
		\label{eq:permutation_update}
	\end{equation}
\end{lemma}
\begin{proof}
	We proof the statement by induction.
	For $s=1$ \cref{eq:label_prod} reduces to
	\begin{equation*}
		\mat{\cycle}^{\alpha}\mat{P}=\left(1-\alpha\right)\mat{P}+\alpha \mat{\cycle}\mat{P}=\mat{P}+\alpha\left(\mat{\cycle}-\mat{I}_n\right)\mat{P}\;,
	\end{equation*}
	leading to \cref{eq:permutation_update}.
	Multiplying the inverse of $\mat{P}$ to the right, we obtain
	\begin{equation*}
		\mat{\cycle}^{\alpha}=\mat{I}_n+\alpha\left(\mat{\cycle}-\mat{I}_n\right)\;.
	\end{equation*}  
	Now, consider $s$ and assume that \cref{eq:permutation_update} holds for $s-1$. Then
	\begin{align*}
		&\left( \prod_{i=1}^{s} \mat{\cycle}_i ^{\alpha_i}\right) 
		=\left( \prod_{i=1}^{s-1} \mat{\cycle}_i ^{\alpha_i}\right)\mat{\cycle}_{s}^{\alpha_{s}} \\
		=&\left( \mat{I}_n+\sum_{i=1}^{s-1}\alpha_i\left(\mat{\cycle}_i-\mat{I}_n\right)\right) \left(\mat{I}_n+\alpha_{s} \left(\mat{\cycle}_{s}-\mat{I}_n\right) \right) \\
		=&\left( \mat{I}_n+\sum_{i=1}^{s}\alpha_i\left(\mat{\cycle}_i-\mat{I}_n\right)\right)
		\left(\sum_{i=1}^{s}\alpha_i\alpha_s\left(\mat{\cycle}_i-\mat{I}_n\right)\left(\mat{\cycle}_s-\mat{I}_n\right)\right)\;,
	\end{align*}
	and it remains to show that $\left(\mat{\cycle}_i-\mat{I}_n\right)\left(\mat{\cycle}_s-\mat{I}_n\right)=\mat{O}$, where $\mat{O}$ is the $n\times n$ matrix consisting only of zeros.
	Since all $\mat{\cycle}_i$ are $2$-cycles, $\mat{\cycle}_i-\mat{I}_n$ only has $4$ non-zeros entries and since $\mat{\cycle}_i$ and $\mat{\cycle}_s$ are disjoint, they have these entries in different rows/columns.
\end{proof}

\cref{eq:permutation_update} can be inserted into \cref{eq:qap_cyclic} to obtain the following \QUBO.
\begin{lemma}
	Assume that $\{\mat{\cycle}_1,\dots,\mat{\cycle}_s\}\subset \mathbb P_n^{(2)}$ is a disjoint set of $2$-cycles, $\mat{P}\in\mathbb P_n$. 
	A \QUBO formulation equivalent to \cref{eq:qap_cyclic} is given by 
	\begin{equation}
		\argmin_{\vec{\alpha}\in\{0,1\}^s}\ \vec{\alpha}^{\top}\tilde{\mat{Q}}\vec{\alpha}\;,
		\label{eq:alpha_qubo}
	\end{equation}
	with
	\begin{equation}
		\tilde{Q}_{ij}\defeq
		\begin{cases}
			\cost{\tilde{\mat{\cycle}}_i,\tilde{\mat{\cycle}}_j}, &\text{if }i\neq j\;, \\
			\cost{\tilde{\mat{\cycle}}_i}+\cost{\tilde{\mat{\cycle}}_i,\mat{P}}+
			\cost{\mat{P},\tilde{\mat{\cycle}}_i}, &\text{else}\;.
		\end{cases}
		\label{eq:alpha_qubo_def}
	\end{equation}
\end{lemma}
\begin{proof}
	Since tr and matrix multiplication are linear functions, $c$ is bilinear and we obtain
	\begin{align*}
		\argmin_{\vec{\alpha}\in\{0,1\}^s}&\quad \cost{g\left(\mat{P},\vec{\alpha},\mathbb{\cycle}\right)} \\
		=\argmin_{\vec{\alpha}\in\{0,1\}^s}&\quad\cost{\mat{P}+\sum_{i=1}^{s}\alpha_i\tilde{\mat{\cycle}}_i,\mat{P}+\sum_{j=1}^{s}\alpha_j\tilde{\mat{\cycle}}_j} \\
		=\argmin_{\vec{\alpha}\in\{0,1\}^s}&\quad\cost{\mat{P}}+\sum_{i,j=1}^{s}\alpha_i\alpha_j\cost{\tilde{\mat{\cycle}}_i,\tilde{\vec{\cycle}}_j} \\
		&+\sum_{i=1}^{s}\alpha_i\cost{\tilde{\mat{\cycle}}_i,\mat{P}}
		+\sum_{j=1}^{s}\alpha_j\cost{\mat{P},\tilde{\vec{\cycle}}_j} 
		=\argmin_{\vec{\alpha}\in\{0,1\}^s} \vec{\alpha}^{\top}\tilde{\mat{Q}}\vec{\alpha}\;,
	\end{align*}
	where $\tilde{\mat{Q}}$ is defined as in \cref{eq:alpha_qubo_def}.
\end{proof}
We observe that for a set with $n$ elements, the largest possible set of disjoint $2$-cycles has $\left\lfloor n/2\right\rfloor$ elements. 
Therefore, the dimension of the \QUBO problem in \eqref{eq:alpha_qubo} is way smaller than the size of the original \QUBO in \eqref{eq:qap_qubo} ($s\le \left\lfloor n/2\right\rfloor\ll m(n+1)$). 

The overall iterative method is outlined in \cref{alg:expansion}: Given a permutation matrix $\mat{P}\in\mathbb P_n$, we iteratively choose sets of random disjoint cycles and optimize \cref{eq:alpha_qubo}.
This gives us a binary vector $\vec{\alpha}$, indicating which cycle should be applied to our current permutation matrix. 
After updating $\mat{P}$, the procedure is repeated until convergence.
\begin{algorithm}[t]
	\caption{\CX Algorithm}\label{alg:expansion}
	\begin{algorithmic}[1]
		\Require $\mat{F}\in\mathbb \mathbb R^{m\times m}$, $\mat{D}\in\mathbb R^{n\times n}$, $k\le m$, $k_u\le n - k$
		\Ensure Sub-permutation matrix $\mat{P}\in\mathbb P_{m,n}$ optimizing $\cost{\mat{P}}$
		\State Initialize $\mat{P}\in\mathbb P_{m,n}$ %(\cref{sec:initialization})
		\label{alg:expansion:initialization}
		\Repeat
		\State Choose $\set{I}\subset [m]$, $\left|\set{I}\right|=k$, $\set{J}\subset [n]$, $\left|\set{J}\right|=k_u$ (\cref{sec:cycles})
		\label{alg:expansion:indices}
		\State Construct matrix $\mat{Q}\left(\set{I},\set{J}\right)$ (\cref{sec:subproblem})
		\label{alg:expansion:matrix}
		\Repeat
		\label{alg:expansion:inner_start}
		\State Choose a random set of $2$-cycles $\mathbb{\cycle}$ (\cref{sec:cycles})
		\label{alg:expansion:cycles}
		\State Calculate matrix $\tilde{\mat{Q}}\left(\set{I},\set{J}\right)$ from $\mat{Q}\left(\set{I},\set{J}\right)$ (\cref{eq:alpha_qubo_def}) 
		\label{alg:expansion:submatrix}
		\State $\vec{\alpha}^*\gets \argmin_{\vec{\alpha}}\vec{\alpha}^{\top}\tilde{\mat{Q}}\left(\set{I},\set{J}\right)\vec{\alpha}$
		\label{alg:expansion:qubo}
		\Comment{QC can be used}
		\State $\mat{P}\gets g\left(\mat{P},\vec{\alpha}^*,\mathbb{\cycle}\right)$ (\cref{eq:permutation_update})
		\label{alg:expansion:update}
		\Until{Every $2$-cycle occured in one set}
		\label{alg:expansion:inner_end}
		\Until{A convergence criterium is met}
		\label{alg:expansion:convergence}
	\end{algorithmic}
\end{algorithm}
The specifics for \cref{alg:expansion:initialization,alg:expansion:indices,alg:expansion:matrix} are elaborated in the next subsections.
%
%\subsubsection{Initialization}\label{sec:initialization}

Since our proposed method works iteratively, any given initial solution can be incorporated easily.
Either we can start off with a random sub-permutation or something more elaborated like analytical or force-directed placement \cite{feld2017fieldplacer}.

\subsubsection{Choosing Indices}\label{sec:indices}

Instead of optimizing over the whole index set $[m]$ in each iteration we can reduce the problem size by considering an index set $\mathcal{I}\subset [m]$ of size $k$.
These indices can be chosen randomly but having a deep understanding of the underlying problem setting one could use a more informative approach.
For example, we could choose the indices depending on the impact of the overall cost, i.e.,
\begin{equation}
	\argmax_{\set{I}\subset [m],\left|\set{I}\right|=k}\sum_{i\in \set{I}}\sum_{j=1}^{n}F_{i,j}D_{\pi(i),\pi(j)}\;. \label{eq:worst}
\end{equation} 
Intuitively, it makes sense to greedily permute the currently worst performing indices. However, one might get stuck in a local optimum too early. 
Both methods (random and greedy) are investigated later on in \cref{sec:experiments}.

Even though we now have a problem dimension only dependent on $k$ and not the number of locations $n$, this approach is not yet guaranteed to converge towards an optimal solution.
If we only consider index sets $\mathcal{I}\subset [m]$ we stick to permute the initial sub-permutation and thus concentrate on a fixed set of $m$ locations. 
To prevent this, in each iteration, we also sample a set $\mathcal{J}\subset [n]$ of $k_u\le k$ unbound locations, i.e., locations which are not assigned to a functional block.
We sample each unbound location with a probability proportional to the distance to the nearest neighbor in the set of bound locations.
With this method, we also explore the set of unbound locations and can place the given functional blocks on the whole chip grid. 

\subsubsection{Choosing Cycles}\label{sec:cycles}

For fixed numbers of indices $k,k_u\le m$, we iteratively sample a set of disjoint $2$-cycles and optimize the current permutation until every $2$-cycle has occurred in the sampling process (Lines \ref{alg:expansion:inner_start}-\ref{alg:expansion:inner_end}).
The question arises on which indices these cycles should be sampled.
Considering an index set $\set{I}=\{i_1,\dots,i_k\}\subset [m]$ we can compute the sub-permuted index set under the sub-permutation $\pi$ as $\set{I}_{\pi}\defeq\{\pi(i_1),\dots,\pi(i_k)\}\subset [n]$.
We first sample $k_u$ disjoint $2$-cycles which map $\{i_1,\dots,i_{k_u}\}$ to $\mathcal{J}$.
Secondly, we sample $\left\lfloor (k-k_u)/2\right\rfloor$ disjoint $2$-cycles which map $\{i_{k_u+1},\dots,i_k\}$ to $\{\pi(i_{k_u+1}),\dots,\pi(i_k)\}$.
Since there are restrictions on which functional block can be mapped to which chip location (e.g. a LUT cannot be placed on an IO cell), one cannot simply sample these cycles arbitrarily between chosen indices (see \cref{fig:constraints}).
\begin{figure}
	\centering
	\includegraphics[width=0.4\textwidth]{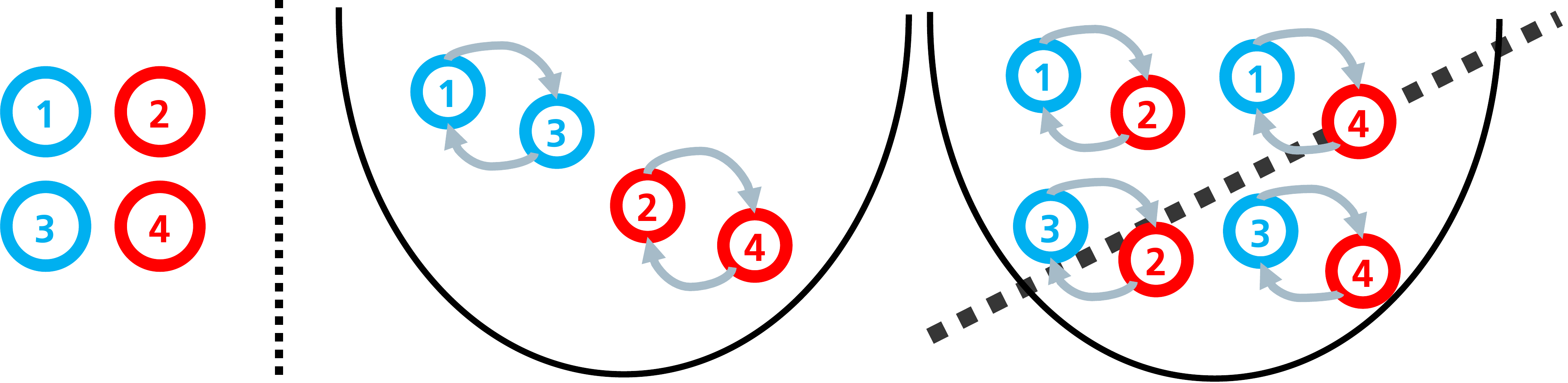}
	\caption{Illustration of allowed 2-cycles. Current placement (left) with corresponding cell types (cyan and red). Exemplary legal and illegal 2-cycles (right).}
	\label{fig:constraints}
\end{figure}
However, this does not pose a problem for this framework, because these constraints can be integrated into the sampling process.
The total number of $2$-cycles is thus $s=k_u+\left\lfloor (k-k_u)/2\right\rfloor$, which is only dependent on the freely selectable index set sizes $k$ and $k_u$.
Since the \QUBO dimension corresponds exactly to the number of considered disjoint cycles, we can conveniently adapt the problem to the available hardware size, either for real quantum annealers or classical digital annealing \QUBO solvers.
However, there is a trade-off between problem size and performance, which will be investigated later on in \cref{sec:experiments}. 

\subsubsection{Constructing $\mat{Q}\left(\set{I},\set{J}\right)$}\label{sec:subproblem}

It remains to clarify how the matrix $\mat{Q}\left(\set{I},\set{J}\right)$ from \cref{alg:expansion:submatrix} is constructed.
One way would be to precompute the cost matrix $\mat{Q}\in\mathbb R^{mn\times mn}$ and then using standard methods for reducing the \QUBO size with fixed variables.
Since real-world algorithms to be implemented on an \FPGA-chip can contain up to millions of functional blocks and chip locations, the size of the cost matrix $\mat{Q}\in\mathbb R^{mn\times mn}$ (\cref{eq:qap_qubo}) can get infeasible to hold the precomputed matrix in memory.
We resolve this issue by exploiting the tensor product-like structure of $\mat{Q}=\mat{F}\otimes\mat{D}$.
\begin{definition}
	For an index set $\set{I}\subset[m]$ and a sub-permutation matrix $\mat{P}\in\mathbb P_{m,n}$, define $\set{I}^c\defeq [m]\setminus \set{I}$ and let $\set{I}_{\pi},\set{I}^c_{\pi}\subset [n]$ be the sets created by applying the underlying sub-permutation $\pi$ of $\mat{P}$ to $\set{I},\set{I}^c$.
\end{definition}
\begin{lemma}
	Given $\mat{F}\in\mathbb R^{m\times m}$, $\mat{D}\in\mathbb R^{n\times n}$, $\mat{P}\in\mathbb P_{m,n}$, $\set{I}\subset [m]$ with $\left|\mathcal I\right|=k$ and $\set{J}\subset [n]$ with $\left|\mathcal I\right|=k_u$ it holds
	\begin{align*}
		&\argmin_{\vec{x}\left(\set{I},\set{J}\right)\in\{0,1\}^{k(k+k_u)}}\vec{x}^{\top}\mat{Q}\vec{x} \\
		=&\argmin_{\vec{x}\left(\set{I},\set{J}\right)\in\{0,1\}^{k(k+k_u)}}\left(\vec{x}\left(\set{I},\set{J}\right)\right)^{\top}\mat{Q}\left(\set{I},\set{J}\right)\vec{x}\left(\set{I},\set{J}\right)\;,
	\end{align*}
	with $\vec{x}\left(\set{I},\set{J}\right)\defeq\vect(\mat{P}_{\set{I},\set{I}_{\pi}'})$, $\set{I}_{\pi}'\defeq\set{I}_{\pi}\cup\set{J}$, $\vec{x}=\vect\left(\mat{P}\right)$ and
	\begin{equation}
		\mat{Q}\left(\set{I},\set{J}\right)\defeq \mat{F}_{\set{I}}\otimes \mat{D}_{\set{I}_{\pi}'}+2\ \diag\left(\vect\left(\mat{F}_{\set{I},\set{I}^c} \mat{D}_{\set{I}_{\pi}^c,\set{I}_{\pi}'}\right)\right)\;.
		\label{eq:q_i}
	\end{equation}
\end{lemma}
\begin{proof}
	With the following equality
	\begin{align*}
		%	&\cost{\mat{P}} \\
		%	=
		&\ \tr \left( \mat{F} \mat{P} \mat{D} \mat{P}^{\top} \right) \\
		%	=&\tr \left( \mat{F} \mat{D}_{[n]_{\mat{P}}} \right)
		%	=&\sum_{i,j\in[n]}F_{i,j}D_{\pi(i),\pi(j)} \\
		%	=&\sum_{i,j\in\set{I}}F_{i,j}D_{\pi(i),\pi(j)}+\sum_{i,j\in\set{I}^c}F_{i,j}D_{\pi(i),\pi(j)} 	\\
		%	+&\sum_{i\in\set{I},j\in\set{I}^c}F_{i,j}D_{\pi(i),\pi(j)}+\sum_{i\in\set{I}^c,j\in\set{I}}F_{i,j}D_{\pi(i),\pi(j)} \\
		=&\ \tr \left( \mat{F}_{\set{I}}\mat{P}_{\set{I},\set{I}_{\pi}'} \mat{D}_{\set{I}} \mat{P}_{\set{I},\set{I}_{\pi}'}^{\top}\right) 
		+\tr \left( \mat{F}_{\set{I}^c}\mat{P}_{\set{I}^c,\set{I}^c_{\pi}} \mat{D}_{\set{I}^c} \mat{P}_{\set{I}^c,\set{I}^c_{\pi}}^{\top}\right)  \\
		+&2\ \tr \left( \mat{F}_{\set{I},\set{I}^c}\mat{P}_{\set{I}^c,\set{I}^c_{\pi}} \mat{D}_{\set{I}^c,\set{I}} \mat{P}_{\set{I},\set{I}_{\pi}'}^{\top}\right)\;,
	\end{align*}
	we can deduce
	\begin{align*}
		&\argmin_{\vec{x}\left(\set{I},\set{J}\right)\in\{0,1\}^{k(k+k_u)}}\vec{x}^{\top}\mat{Q}\vec{x} 
		 \Leftrightarrow
		\argmin_{\mat{P}_{\set{I},\set{I}_{\pi}'}\in\mathbb P_{k,k+k_u}}\cost{\mat{P}} \\
		=&\argmin_{\mat{P}_{\set{I},\set{I}_{\pi}'}\in\mathbb P_{k,k+k_u}}\tr \left( \mat{F}_{\set{I}}\mat{P}_{\set{I},\set{I}_{\pi}'} \mat{D}_{\set{I}} \mat{P}_{\set{I},\set{I}_{\pi}'}^{\top}\right) \\
		&\qquad +2\ \tr \left( \mat{F}_{\set{I},\set{I}^c}\mat{P}_{\set{I}^c,\set{I}^c_{\pi}} \mat{D}_{\set{I}^c,\set{I}} \mat{P}_{\set{I},\set{I}_{\pi}'}^{\top}\right)\\
		\Leftrightarrow&\argmin_{\vec{x}\left(\set{I},\set{J}\right)\in\{0,1\}^{k(k+k_u)}}\left(\vec{x}\left(\set{I},\set{J}\right)\right)^{\top} \left(\mat{F}_{\set{I}}\otimes \mat{D}_{\set{I}_{\pi}'}\right)\vec{x}\left(\set{I},\set{J}\right) \\
		&\qquad+2\ \left(\vect\left(\mat{F}_{\set{I},\set{I}^c} \mat{D}_{\set{I}_{\pi}^c,\set{I}_{\pi}'}\right)\right)^{\top}\vec{x}\left(\set{I},\set{J}\right) \\
		=&\argmin_{\vec{x}\left(\set{I},\set{J}\right)\in\{0,1\}^{k(k+k_u)}}\left(\vec{x}\left(\set{I}\right)\right)^{\top}\mat{Q}\left(\set{I},\set{J}\right)\vec{x}\left(\set{I},\set{J}\right)\;,
	\end{align*}
	with $\mat{Q}\left(\set{I},\set{J}\right)$ as defined in \cref{eq:q_i}.
\end{proof}

With having clarified all steps of \cref{alg:expansion}, we can have a look at the behavior of this algorithm.

\section{Experiments} \label{sec:experiments}

We conduct experiments with a fictional \FPGA architecture for analyzing the behavior of our proposed algorithm.
We choose this as a generic minimum baseline for all FPGA architectures, ignoring implementation details like grouping into slices, carry chains etc, which might be vendor specific and thus not translate easily to other FPGAs. 
In this architecture, we assume that every LUT has an adjacent register, so that its usage is irrelevant to the placement process and can be ignored. 
We consider only the data path, i.e., ignore clock net routing and control signals like reset and clock enable.
However, this information can be integrated into $\mat{F}$ and $\mat{D}$.

The fictional \FPGA architecture contains three different cell types: IO cells, BRAM cells and LUT cells.
The legend for upcoming plots is indicated in \cref{fig:legend}.
\begin{figure}[t]
	\centering
	\includegraphics[width=0.48\textwidth]{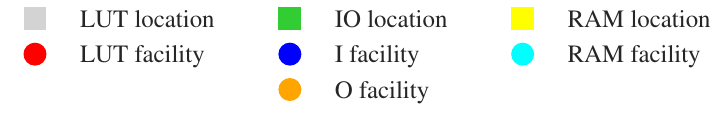}
	\caption{Different cell types for our fictional \FPGA architecture along with plotted colors.}
	\label{fig:legend}
\end{figure} 
For the upcoming experiments we assume an \FPGA chip which consists of $21\times 21$ cells. 
It contains IO cells at the border and 16 BRAM cells distributed uniformly over the gird, with the rest being LUT cells (see e.g. \cref{fig:place100}).
We are thus faced with $n=21^2=441$ locations.

\subsection{Generic Examples}\label{sec:generic_examples}

For examining the behavior of \cref{alg:expansion}, we sample 10 different problem instances with $m=100$ and $m=200$ facilities, respectively.
%The underlying net structure, i.e., which facility should be connected to which other facility, is assumed to have a tree structure.
%That is, the adjacency graph of every instance corresponds to a random tree consisting of $m$ nodes.
For every instance we assume two IO cells, imitating a single input and a single output cell.
The rest of the cell types are randomly sampled corresponding to the ratio of the underlying architecture.

We compare the performance of \cref{alg:expansion} with solving the \QUBO given in \cref{eq:qap_qubo} using different \QUBO solvers.
As a classical software solver we use a simulated annealing (SA) implemented in the python software package D-Wave Ocean (\url{https://docs.ocean.dwavesys.com/en/stable}) with default parameters.
As a second classical solver, we utilize a \QUBO hardware solver, which is denoted as digital annealing (DA). 
Similar to QA, DA is standalone but is not based on quantum technology and uses classical algorithms.
One can set up the running time/annealing time of this device and we henceforth set this time to $0.1~s$ which is equivalent to evaluating $\approx 160k$ candidate solutions.
%\TG{@Nico: What reference to use here?}
Thirdly, we use a real quantum annealer (QA), namely a D-Wave Advantage System 5.4 with 5614 qubits and 40,050 couplers, fixing the annealing time to $40~\mu s$ and taking the best out of $100$ reads.
Since our algorithm \CX contains random decisions, such as choosing a set of cycles in Line \ref{alg:expansion:cycles}, we plot the average performance over 10 runs and indicate the $95\%$-confidence intervals.

We start with depicting the performance of \CX over 50 iterations in terms of the \QAP cost in \cref{fig:generic_k}, varying the number of chosen unbound indices $k_u$.
We fix $k=100$ and compare $k_u\in\{10, 50, 100\}$ using the SA solver, with the performance being averaged over the 10 generated instances for $m=100$ and $m=200$.
Moreover, the impact of different methods for choosing sub-problems in Line \ref{alg:expansion:indices} is indicated, comparing random sampling with worst performing indices \cref{eq:worst}.
\begin{figure*}[t]
	\centering
	\begin{subfigure}{0.44\textwidth}
		\centering
		\includegraphics[width=1\textwidth]{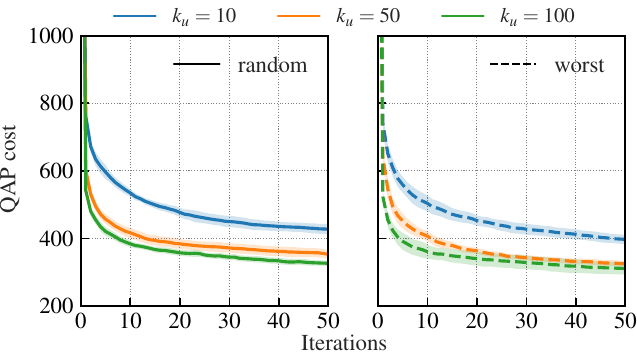}
		\caption{$m=100$.}
		\label{fig:generic_k:100}
	\end{subfigure}
	\hspace{0.04\textwidth}
	\begin{subfigure}{0.44\textwidth}
		\centering
		\includegraphics[width=1\textwidth]{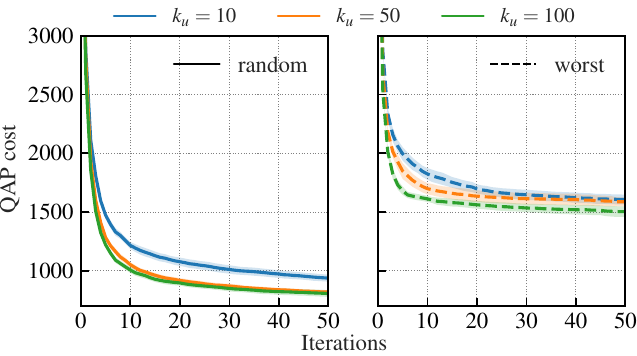}
		\caption{$m=200$.}
		\label{fig:generic_k:200}
	\end{subfigure}
	\caption{Depicting the effect of varying $k_u$ when $k$ is fixed to a certain value, comparing choosing random sub-problems in Line \cref{alg:expansion:indices} with choosing worst performing indices \cref{eq:worst}. Here, $k=100$ and the \QAP cost is depicted over 50 iterations used in the \CX. We compare the costs for 10 randomly generated problems with a problem size of 100 (a) with 10 problems of size 200 (b). }
	\label{fig:generic_k}
\end{figure*}
We observe that the \QAP cost decreases with every iteration of the algorithm.
For $m=100$ the random indices choosing method performs similar to the method of worst indices, contrary to the case $m=200$.
The larger the dimension of our problem gets, considering only the worst indices can lead to fast convergence to local optima, leading to an overall worse placement in the end.
Furthermore, we can see that with an increasing number of unbound variables $k_u$, the performance of the \CX increases, since the space of unbound cells is more thoroughly explored.
However, increasing this parameter $k_u$ also leads to a larger \QUBO size (\cref{eq:alpha_qubo}), leading to a trade-off between problem size and performance for a fixed number of iterations.

An experiment with similar configuration can be found in \cref{fig:generic_k_u}, but we know compare the performance varying the dimension of the chosen sub-problems $k$. 
\begin{figure*}[t]
	\centering
	\begin{subfigure}{0.44\textwidth}
		\centering
		\includegraphics[width=1\textwidth]{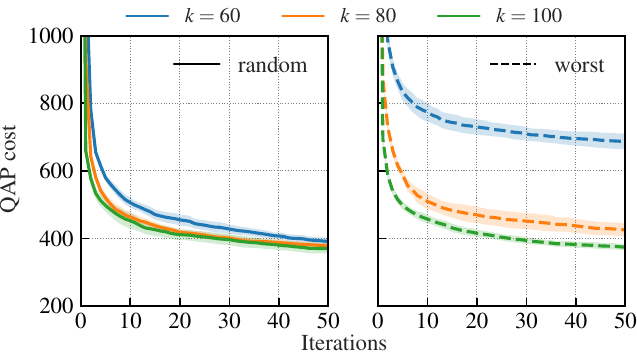}
		\caption{$m=100$.}
		\label{fig:generic_k_u:100}
	\end{subfigure}
	\hspace{0.04\textwidth}
	\begin{subfigure}{0.44\textwidth}
		\centering
		\includegraphics[width=1\textwidth]{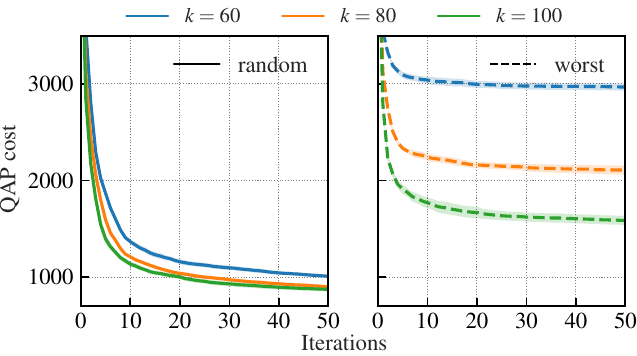}
		\caption{$m=200$.}
		\label{fig:generic_k_u:200}
	\end{subfigure}
	\caption{Depicting the effect of varying $k$ when $k_u$ is fixed to a certain value, comparing choosing random sub-problems in Line \cref{alg:expansion:indices} with choosing worst performing indices \cref{eq:worst}. Here, $k_u=30$ and the \QAP cost is depicted over 50 iterations used in the \CX. We compare the costs for 10 randomly generated problems with a problem size of 100 (a) with 10 problems of size 200 (b). }
	\label{fig:generic_k_u}
\end{figure*}
Here, the number of unbound indices is fixed to $k_u=30$.
We observe similar behavior as for \cref{fig:generic_k} but choosing the worst sub-indices especially falls back in performance to choosing random indices for small $k$ and large $m$.

Fixing $k=100$ and $k_u=50$, we plot intermediate placement results of \cref{alg:expansion} after $1$, $10$ and $50$ iterations in \cref{fig:place100,fig:place200}.
We fix the locations of the IO cells before the actual placement and use a random initialization.
\begin{figure*}[t]
	\centering
	\begin{subfigure}{0.22\textwidth}
		\centering
		\includegraphics[width=1\textwidth]{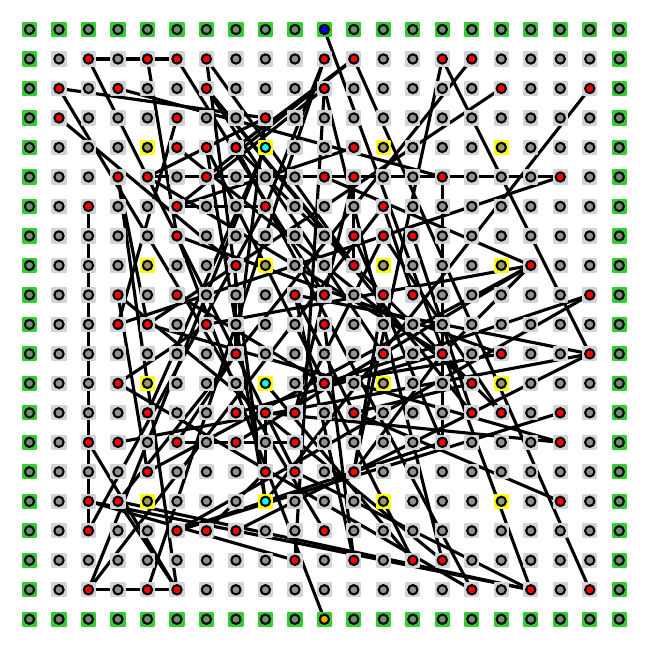}
		\caption{Initial placement, cost 2492}
		\label{fig:place100:0}
	\end{subfigure}
	\hspace{0.02\textwidth}
	\begin{subfigure}{0.22\textwidth}
		\centering
		\includegraphics[width=1\textwidth]{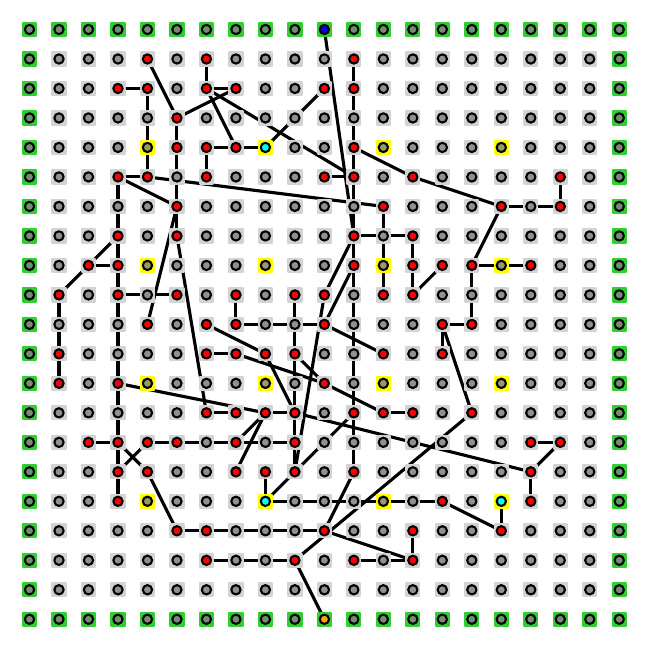}
		\caption{After 1 iteration, cost 576}
		\label{fig:place100:1}
	\end{subfigure}
	\hspace{0.02\textwidth}
	\begin{subfigure}{0.22\textwidth}
		\centering
		\includegraphics[width=1\textwidth]{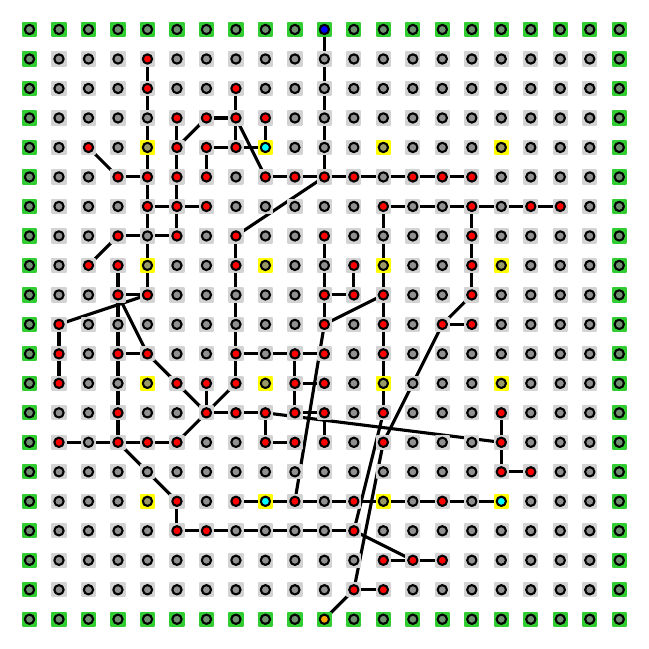}
		\caption{After 10 iterations, cost 388}
		\label{fig:place100:10}
	\end{subfigure}
	\hspace{0.02\textwidth}
	\begin{subfigure}{0.22\textwidth}
		\centering
		\includegraphics[width=1\textwidth]{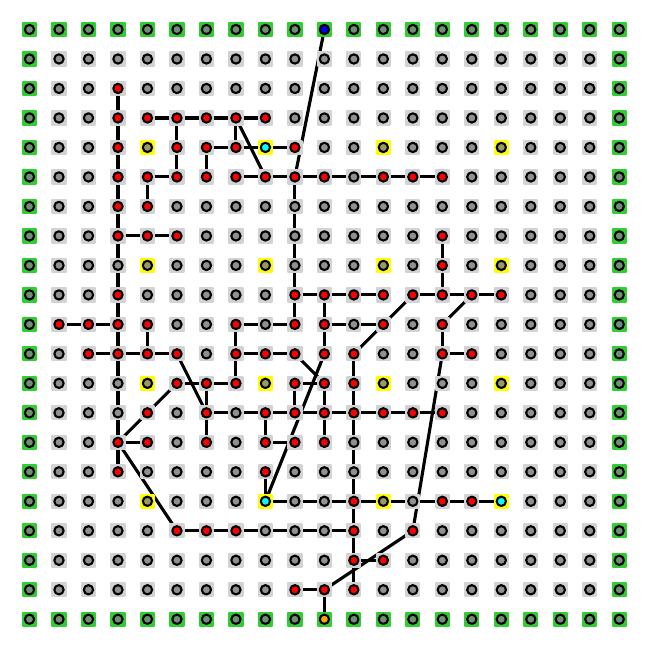}
		\caption{After 50 iterations, cost 328}
		\label{fig:place100:50}
	\end{subfigure}
	\caption{Intermediate placement results for an exemplary generic example with 100 facilities. The initial random placement (a) is indicated along with the result of applying our algorithm for 1 iteration (b), 10 iterations (c) and 50 iterations (d). The placement of the two IO facilities is fixed and the corresponding \QAP costs are indicated. See \cref{fig:legend} for a legend.}
	\label{fig:place100}
\end{figure*}
We can see that the initial placement is very bad, in the sense that the connecting edges are spread over the whole chip grid and cross each other, leading to a large \QAP cost.
With an increasing number of iterations, the placement gets a more grid-like structure with less crossings, leading to very preferable results for a potential subsequent routing.
In \cref{fig:place200}, the intermediate placements for a random instance with $m=200$ and fixed IO locations is shown.
We can see, that it takes more iterations to achieve a good placement on the first sight, than for $m=100$.
\begin{figure*}[t]
	\centering
	\begin{subfigure}{0.22\textwidth}
		\centering
		\includegraphics[width=1\textwidth]{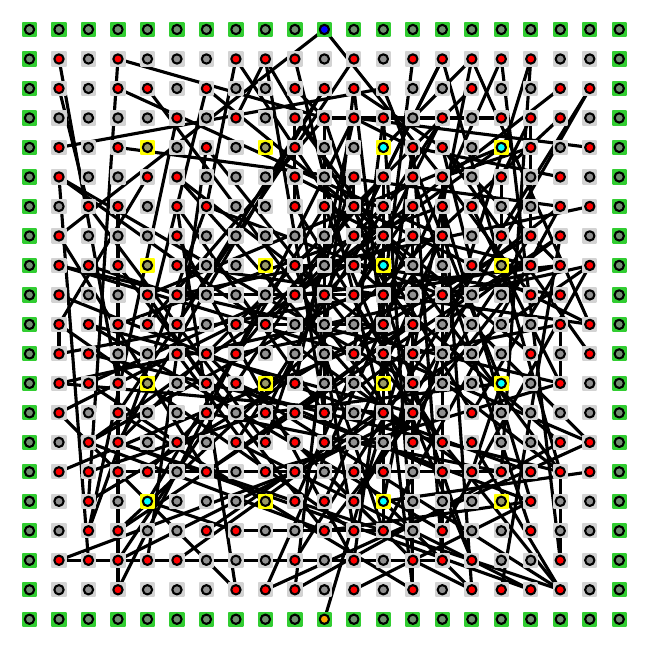}
		\caption{Initial placement, cost 4954}
		\label{fig:place200:0}
	\end{subfigure}
	\hspace{0.02\textwidth}
	\begin{subfigure}{0.22\textwidth}
		\centering
		\includegraphics[width=1\textwidth]{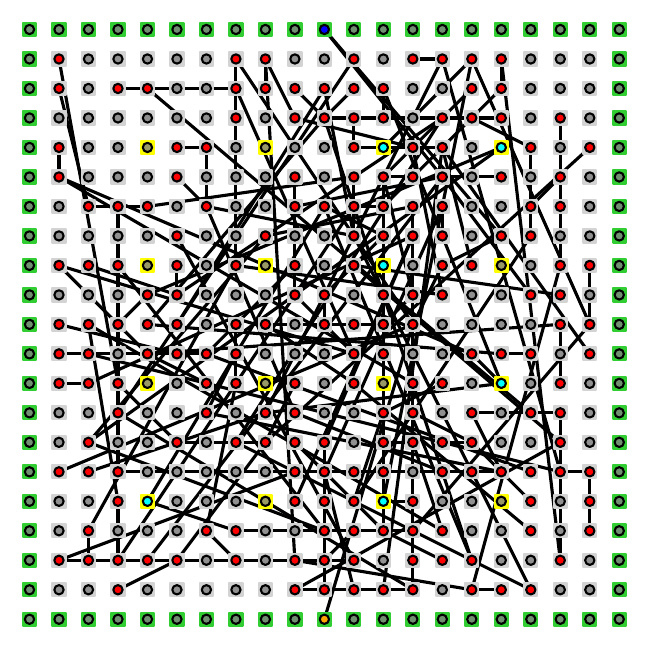}
		\caption{After 1 iteration, cost 2780}
		\label{fig:place200:1}
	\end{subfigure}
	\hspace{0.02\textwidth}
	\begin{subfigure}{0.22\textwidth}
		\centering
		\includegraphics[width=1\textwidth]{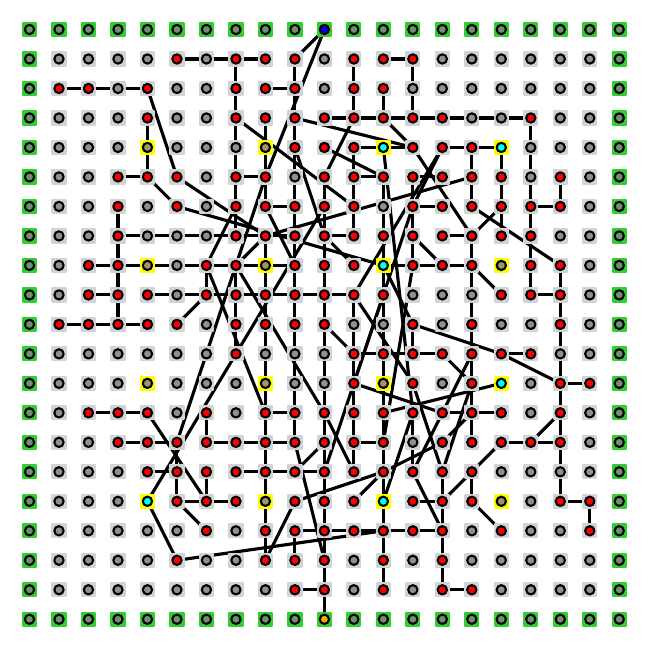}
		\caption{After 10 iterations, cost 1068}
		\label{fig:place200:10}
	\end{subfigure}
	\hspace{0.02\textwidth}
	\begin{subfigure}{0.22\textwidth}
		\centering
		\includegraphics[width=1\textwidth]{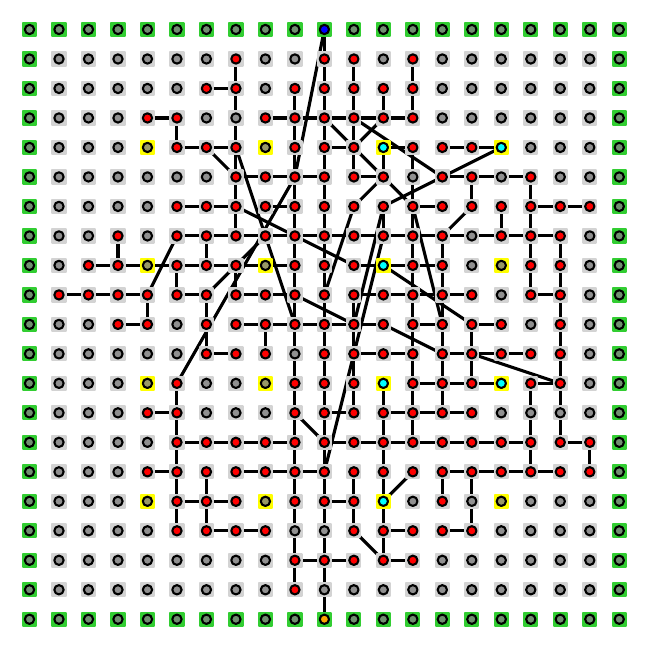}
		\caption{After 50 iterations, cost 768}
		\label{fig:place200:50}
	\end{subfigure}
	\caption{Intermediate placement results for an exemplary generic example with 200 facilities. The initial random placement (a) is indicated along with the result of applying our algorithm for 1 iteration (b), 10 iterations (c) and 50 iterations (d). The placement of the two IO facilities is fixed and the corresponding \QAP costs are indicated. See \cref{fig:legend} for a legend.}
	\label{fig:place200}
\end{figure*}

\subsection{CRC Example}\label{sec:crc}

As a real-world example, we consider a simple 32-bit Cyclic Redundancy Check (CRC-32) algorithm with 8-bit parallel input. This is synthesized by the open-source tool Yosys ({\url{https://yosyshq.net/yosys/}) 0.33 for the Lattice MachXO2 architecture. Wide LUTs and CCU2 carry chains are forbidden, so that the synthesized output only contains 78 LUT-4s and 32 registers. The Lattice MachXO2 is a current technology that is available in sizes starting from 256 LUTs\cite{latticexo2}, so it is comparable to our demo architecture.

For this real-world example, we conduct experiments with a similar configuration to the one used in \cref{fig:generic_k}.
In contrast to the previous experiments, we now do not vary our problem size $m$ but compare setups with fixed IO cells with the setup of optimizing the placement of these cells along with the remaining functional blocks.
In \cref{fig:crc_k}, we fix $k=100$ and vary $k_u\in\{10,50,100\}$.
\begin{figure*}[t]
	\centering
	\begin{subfigure}{0.44\textwidth}
		\centering
		\includegraphics[width=1\textwidth]{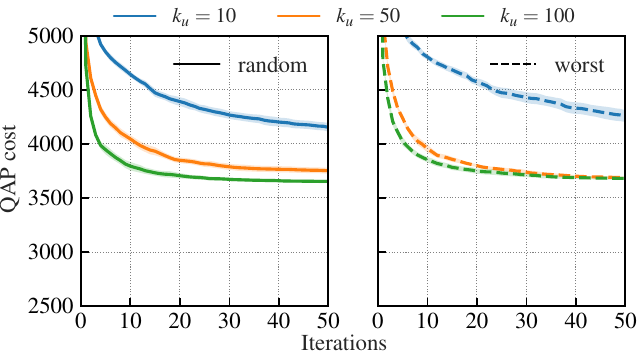}
		\caption{Fixed IO cells.}
		\label{fig:crc_k:fixed}
	\end{subfigure}
	\hspace{0.04\textwidth}
	\begin{subfigure}{0.44\textwidth}
		\centering
		\includegraphics[width=1\textwidth]{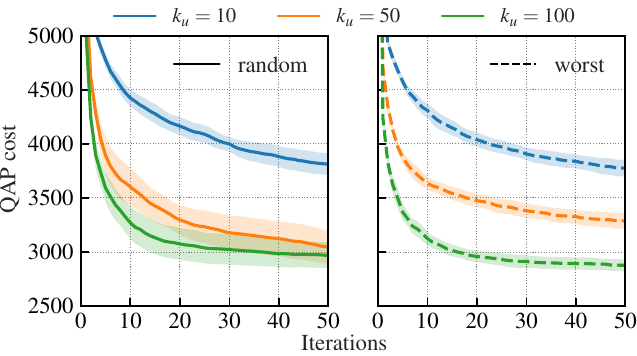}
		\caption{Unfixed IO cells.}
		\label{fig:crc_k:unfixed}
	\end{subfigure}
	\caption{Depicting the effect of varying $k$ when $k_u$ is fixed to a certain value, comparing choosing random sub-problems in Line \cref{alg:expansion:indices} with choosing worst performing indices \cref{eq:worst}. Here, $k_u=30$ and the \QAP cost is depicted over 50 iterations used in the \CX. We compare the costs for fixed IO cells (a) and unfixed IO cells (b) for the CRC-32.}
	\label{fig:crc_k}
\end{figure*}
We observe that fixing the IO cells leads to faster convergence and thus a worse placement than with the possibility of also optimizing the IO placements.
Although, for unfixed IO cells, the uncertainty in the outcomes is larger than in the fixed case.
Furthermore, we can see that the relative performance of only choosing $k_u=10$ unbound indices compared to $k_u=50,100$ is worse than in the generic case (cf. \cref{fig:generic_k}).
The number of needed unbound indices is thus heavily dependent on the underlying problem structure.

In \cref{fig:crc_k_u} we depict the effect of varying $k$ when $k_u$ is fixed to 30.
\begin{figure*}[t]
	\centering
	\begin{subfigure}{0.44\textwidth}
		\centering
		\includegraphics[width=1\textwidth]{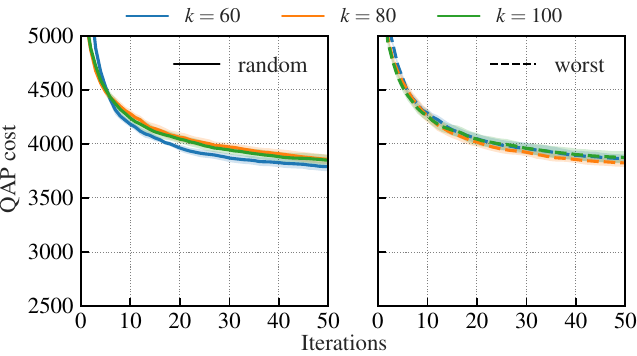}
		\caption{Fixed IO cells.}
		\label{fig:crc_k_u:fixed}
	\end{subfigure}
	\hspace{0.04\textwidth}
	\begin{subfigure}{0.44\textwidth}
		\centering
		\includegraphics[width=1\textwidth]{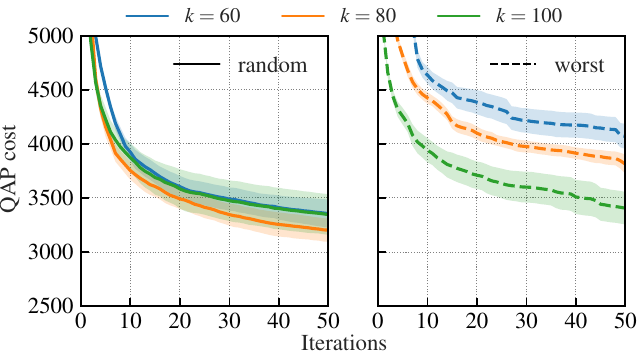}
		\caption{Unfixed IO cells.}
		\label{fig:crc_k_u:unfixed}
	\end{subfigure}
	\caption{Depicting the effect of varying $k$ when $k$ is fixed to a certain value, comparing choosing random sub-problems in Line \cref{alg:expansion:indices} with choosing worst performing indices \cref{eq:worst}. Here, $k=100$ and the \QAP cost is depicted over 50 iterations used in the \CX. We compare the costs for fixed IO cells (a) and unfixed IO cells (b) for the CRC-32. }
	\label{fig:crc_k_u}
\end{figure*}
We observe that for choosing random problems, changing the sub-problem size does not have a very large effect on the \QAP cost.
Thus, one already can achieve good placement results with a small problem size.
However, choosing the sub-problems greedily with \cref{eq:worst} is more sensitive in terms of performance outcomes for different problem sizes.

Intermediate placements after 1, 10 and 50 iterations for fixed and unfixed IO cells can be found \cref{fig:place_fixed,fig:place_unfixed}.
\begin{figure*}[t]
	\centering
	\begin{subfigure}{0.22\textwidth}
		\centering
		\includegraphics[width=1\textwidth]{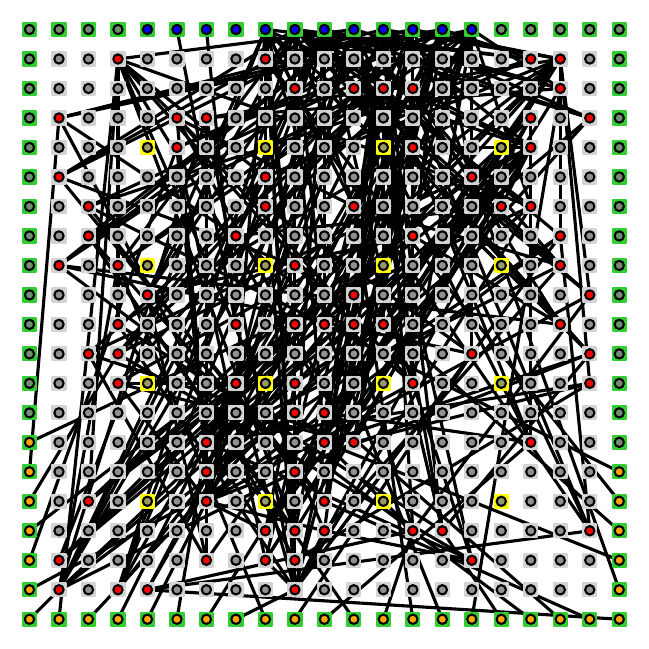}
		\caption{Initial placement, cost 9372}
		\label{fig:place_fixed:0}
	\end{subfigure}
	\hspace{0.02\textwidth}
	\begin{subfigure}{0.22\textwidth}
		\centering
		\includegraphics[width=1\textwidth]{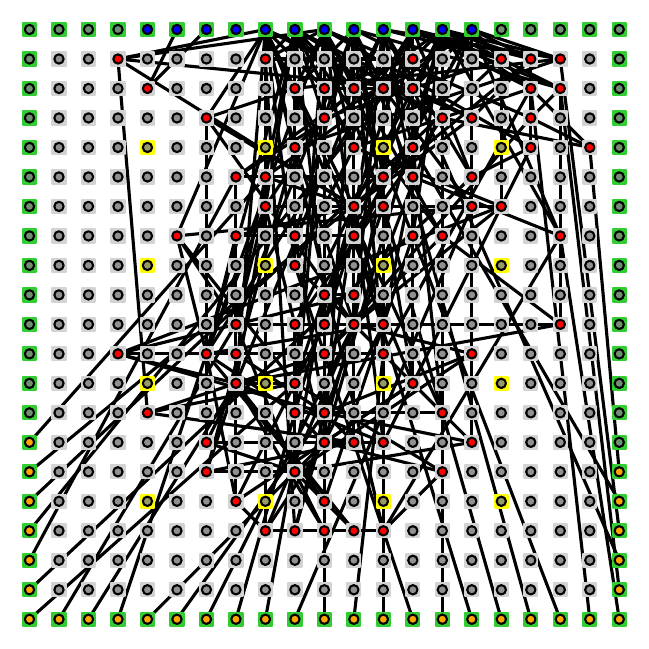}
		\caption{After 1 iteration, cost 4936}
		\label{fig:place_fixed:1}
	\end{subfigure}
	\hspace{0.02\textwidth}
	\begin{subfigure}{0.22\textwidth}
		\centering
		\includegraphics[width=1\textwidth]{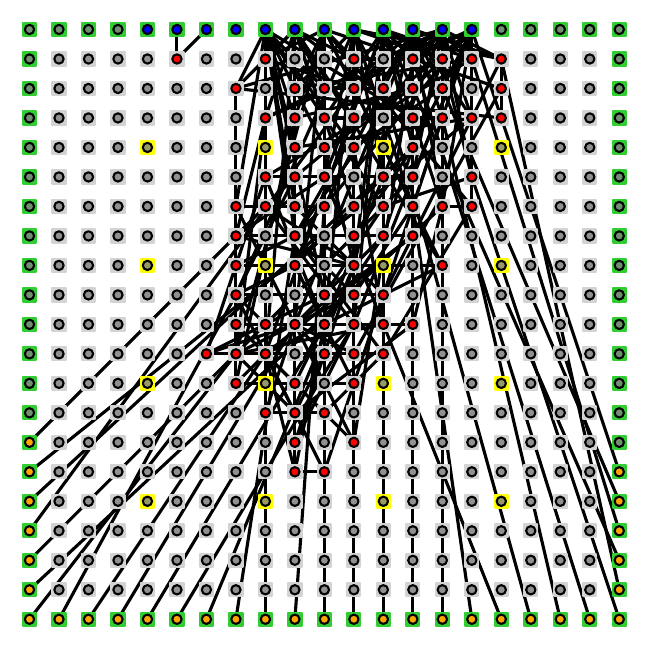}
		\caption{After 10 iterations, cost 4014}
		\label{fig:place_fixed:10}
	\end{subfigure}
	\hspace{0.02\textwidth}
	\begin{subfigure}{0.22\textwidth}
		\centering
		\includegraphics[width=1\textwidth]{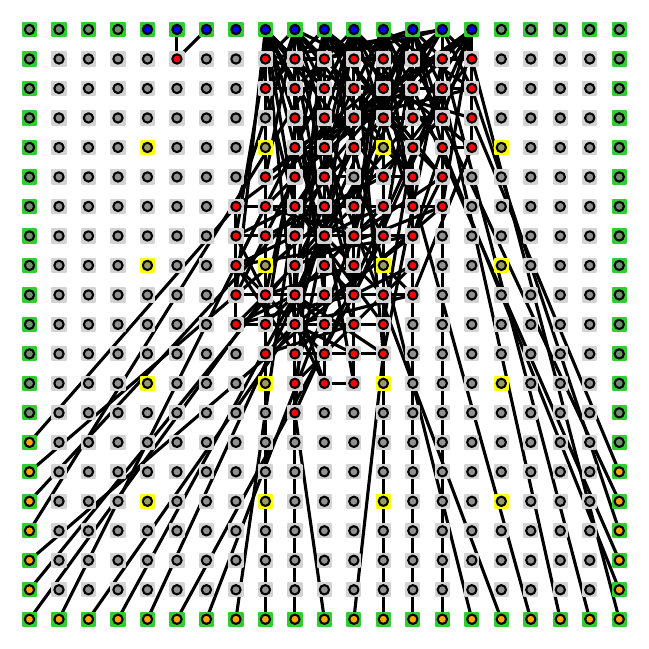}
		\caption{After 50 iterations, cost 3702}
		\label{fig:place_fixed:50}
	\end{subfigure}
	\caption{Intermediate placement results for the CRC-32 example. The initial random placement (a) is indicated along with the result of applying our algorithm for 1 iteration (b), 10 iterations (c) and 50 iterations (d). The placement of the IO facilities is unfixed and optimized. The corresponding \QAP costs are indicated. See \cref{fig:legend} for a legend.}
	\label{fig:place_fixed}
\end{figure*}
\begin{figure*}[t]
	\centering
	\begin{subfigure}{0.22\textwidth}
		\centering
		\includegraphics[width=1\textwidth]{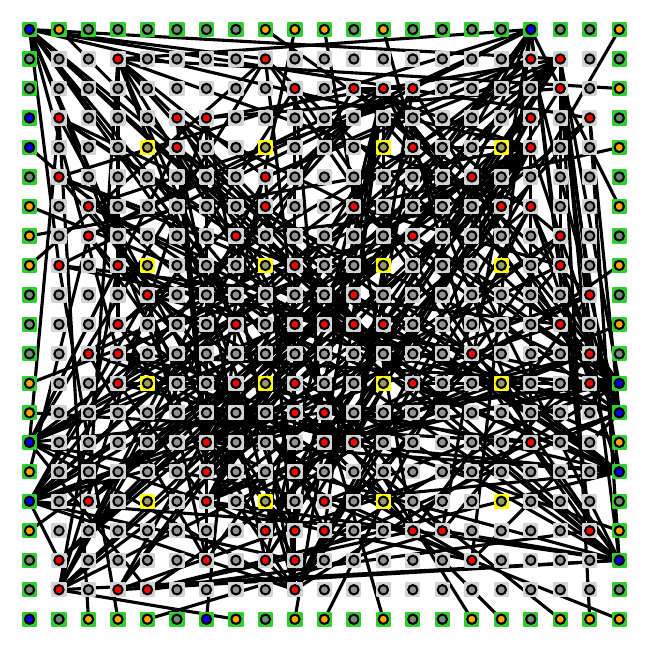}
		\caption{Initial placement, cost 9986}
		\label{fig:place_unfixed:0}
	\end{subfigure}
	\hspace{0.02\textwidth}
	\begin{subfigure}{0.22\textwidth}
		\centering
		\includegraphics[width=1\textwidth]{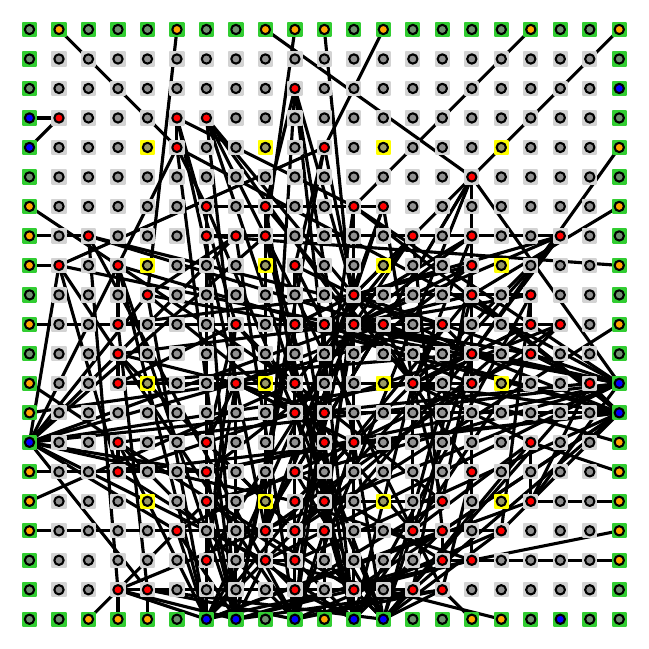}
		\caption{After 1 iteration, cost 5526}
		\label{fig:place_unfixed:1}
	\end{subfigure}
	\hspace{0.02\textwidth}
	\begin{subfigure}{0.22\textwidth}
		\centering
		\includegraphics[width=1\textwidth]{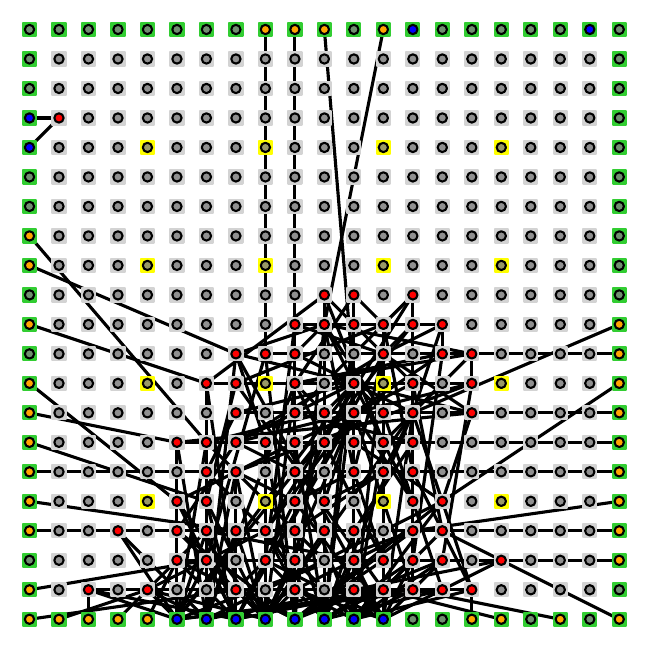}
		\caption{After 10 iterations, cost 3260}
		\label{fig:place_unfixed:10}
	\end{subfigure}
	\hspace{0.02\textwidth}
	\begin{subfigure}{0.22\textwidth}
		\centering
		\includegraphics[width=1\textwidth]{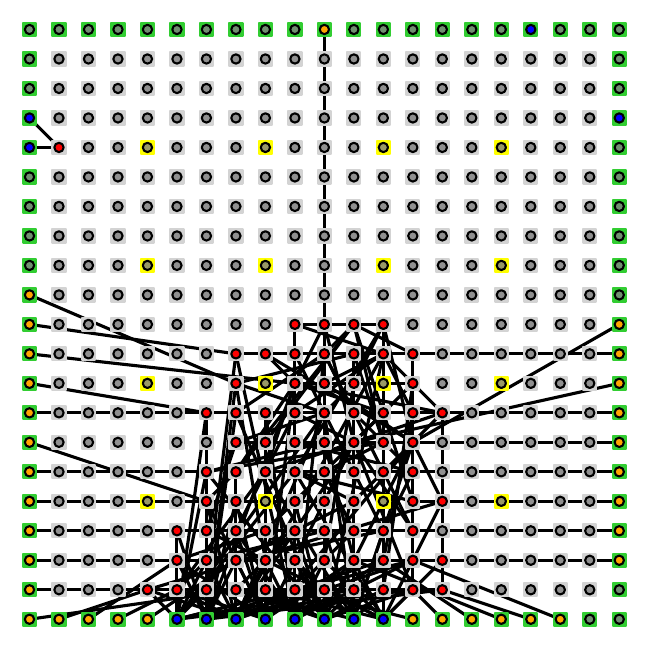}
		\caption{After 50 iterations, cost 2818}
		\label{fig:place_unfixed:50}
	\end{subfigure}
	\caption{Intermediate placement results for the CRC-32 example. The initial random placement (a) is indicated along with the result of applying our algorithm for 1 iteration (b), 10 iterations (c) and 50 iterations (d). The placement of the IO facilities is fixed and the corresponding \QAP costs are indicated. See \cref{fig:legend} for a legend.}
	\label{fig:place_unfixed}
\end{figure*}
Again, similarly to the observations in \cref{fig:place100,fig:place200}, we see that the placement also visually improves with an increasing number of iterations.
When the placement of the IO cells is also optimized, all IO cells are placed close by, leading to an overall better placement than in the fixed case.

Lastly, we conduct experiments with QA and DA. 
We compare the performance of these two hardware solvers with SA on the CRC-32 example with fixed IO cells in \cref{fig:dwave_evo}.
\begin{figure*}
	\centering
	\begin{subfigure}{0.22\textwidth}
		\centering
		\includegraphics[width=1\textwidth]{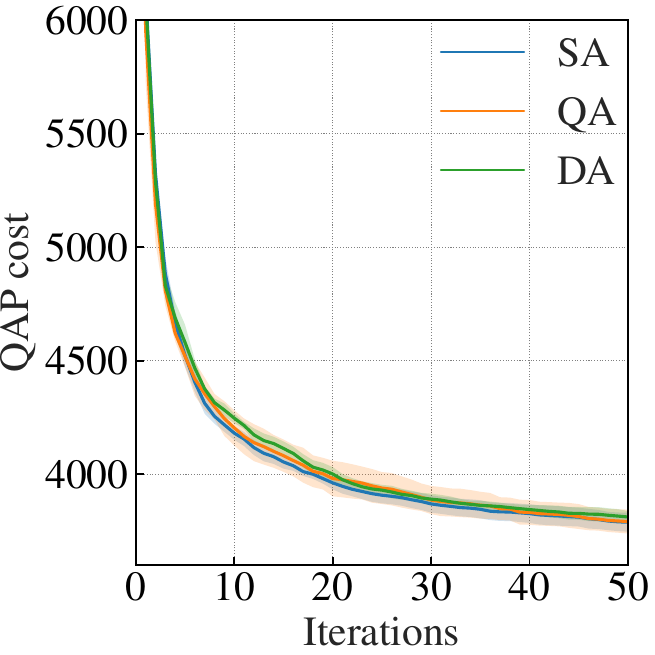}
		\caption{\QAP cost.}
		\label{fig:dwave_evo:qap}
	\end{subfigure}
	\hspace{0.02\textwidth}
	\begin{subfigure}{0.22\textwidth}
		\centering
		\includegraphics[width=1\textwidth]{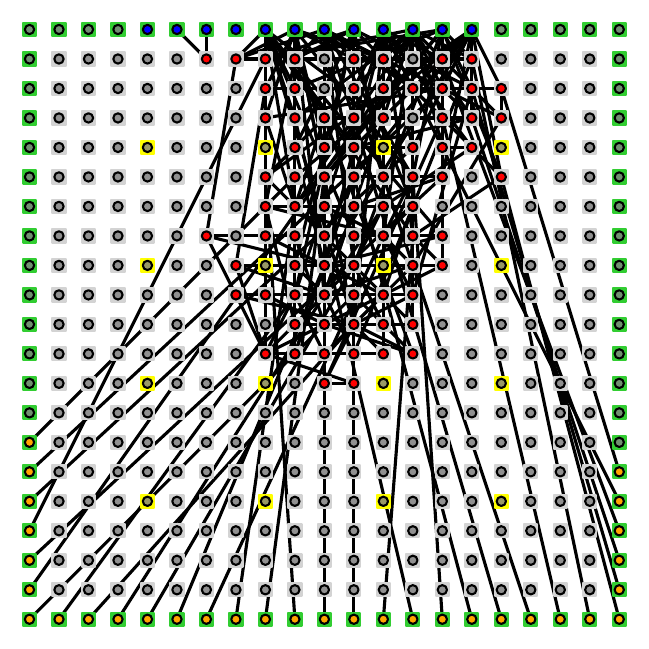}
		\caption{QA, cost 3732.}
		\label{fig:dwave_evo:qa}
	\end{subfigure}
	\hspace{0.02\textwidth}
	\begin{subfigure}{0.22\textwidth}
		\centering
		\includegraphics[width=1\textwidth]{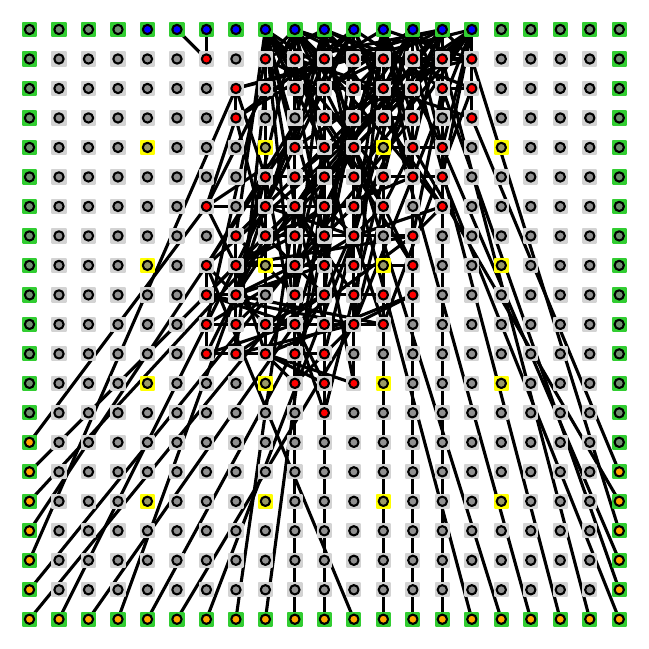}
		\caption{DA, cost 3790.}
		\label{fig:dwave_evo:evo}
	\end{subfigure}
	\hspace{0.02\textwidth}
	\begin{subfigure}{0.22\textwidth}
		\centering
		\includegraphics[width=1\textwidth]{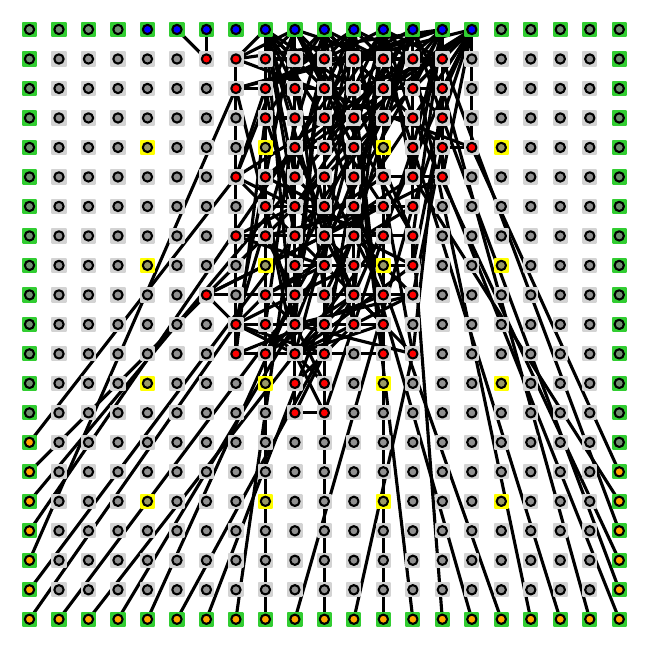}
		\caption{SA, cost 3758.}
		\label{fig:dwave_evo:sa}
	\end{subfigure}
	\caption{Performance comparison of the hardware solvers QA and DA with SA on the CRC-32 example. We choose random sub-problems and fix $k=60$ and $k=30$. We depict the \QAP cost over 50 iterations for the \CX (a) and exemplary placements after 50 iterations with QA (b), DA (c) and SA (d).}
	\label{fig:dwave_evo}
\end{figure*}
We fix $k=60$, $k_u=30$ and choose the sub-problems randomly. 
\cref{fig:dwave_evo:qap} depicts the change of the \QAP cost over an increasing number of iterations in the \CX.
We find all solvers to perform equally well for this problem.
%Intermediate placements after 1, 2 and 5 iterations are depicted in \cref{fig:place_dwave}.
%\begin{figure*}[t]
%	\centering
%	\begin{subfigure}{0.22\textwidth}
%		\centering
%		\includegraphics[width=1\textwidth]{img/crc/dwave/placement_k60_ku30run1_it0.pdf}
%		\caption{Initial placement, cost 9372}
%		\label{fig:place_dwave:0}
%	\end{subfigure}
%	\begin{subfigure}{0.22\textwidth}
%		\centering
%		\includegraphics[width=1\textwidth]{img/crc/dwave/placement_k60_ku30run1_it1.pdf}
%		\caption{After 1 iteration, cost 6140}
%		\label{fig:place_dwave:1}
%	\end{subfigure}
%	\begin{subfigure}{0.22\textwidth}
%		\centering
%		\includegraphics[width=1\textwidth]{img/crc/dwave/placement_k60_ku30run1_it2.pdf}
%		\caption{After 2 iterations, cost 5414}
%		\label{fig:place_dwave:2}
%	\end{subfigure}
%	\begin{subfigure}{0.22\textwidth}
%		\centering
%		\includegraphics[width=1\textwidth]{img/crc/dwave/placement_k60_ku30run1_it5.pdf}
%		\caption{After 5 iterations, cost 4418}
%		\label{fig:place_dwave:5}
%	\end{subfigure}
%	\caption{Intermediate placement results for the CRC32 example using QA. The initial random placement (a) is indicated along with the result of applying our algorithm for 1 iteration (b), 2 iterations (c) and 5 iterations (d). The placement of the IO facilities is fixed and the corresponding \QAP costs are indicated. See \cref{fig:legend} for a legend.}
%	\label{fig:place_dwave}
%\end{figure*}
%We observe similar results to using classical solvers, making quantum devices already viable for solving this problem.

\cref{fig:dwave_evo:qa,fig:dwave_evo:evo,fig:dwave_evo:sa} depict the placement results for specific runs after 50 iterations using QA, DA and SA, respectively.
Similarly good results to \cref{fig:place_fixed:50} can be observed, while the problem sizes are smaller.
Since the \QUBO problems in \eqref{eq:alpha_qubo} are well conditioned (integer valued and small dynamic ranges), real quantum hardware can achieve similar performance to classical solvers. This is an interesting result, since today's quantum technology is still in its infancy with limited computational power (number of qubits) and large proneness to errors.
A detailed discussion on the effect of the conditioning of \QUBO problems can be found in \cite{mucke2023optimum}. 

\section{Conclusion} \label{sec:conclusion}

In this paper, we showed that the \FPGA-\place problem can be solved with quantum computing.
For this, we first formulated the problem in an unbalanced \QAP framework, and then proceeded by working out \QUBO formulations, which can be conveniently solved with quantum annealing and classical hardware solvers.
The \QAP belongs to the hardest problems in NP, meaning there does not exist an approximate algorithm solving this problem in polynomial time. 
%Investigating this problem with the potential is thus a very promising application.
With the notion of sub-permutations, we find a new \QUBO formulation for the unbalanced \QAP without introducing dummy facilities, leading to a lower dimensionality.
The problem of incorporating constraints for the set of allowed sub-permutations is overcome by considering the iterative \CX algorithm.
It works by optimizing over sets of disjoint $2$-cycles, where the size of these sets can be conveniently adapted to the available hardware size.
Moreover, initial placements can easily be incorporated into this algorithm.

Experiments on a fictional \FPGA architecture leads to supporting the theoretical guarantees and show good results.
In this work, we consider binary flow matrices and Manhattan distances between locations on the \FPGA chip.
However, real architectures can easily be integrated into our framework, by adapting the distance matrix and flow matrix correspondingly.
These experiments are conducted on randomly generated problems as well as a small real-world circuit (CRC-32).
The \QUBO problems are solved using software solvers, classical hardware solvers, and real-world quantum annealers. 
We defer the investigation of large-scale use-cases as well as additional performance metrics to follow-up work. 
Here, we were more interested in theoretical properties. To this end, we investigated the effect of varying the sub-problem size $k$ and the number of unbound variables $k_u$, as well as the impact of different methods for choosing the sub-problems on the \QAP cost.
We find that our method allows us to trade-off the solution quality against the running time, 
%Neglecting randomness in choosing the sub-problems can lead to converging to local optima pretty fast, so in addition to conveniently adapting the problem size, one may trade solution quality against running time: either a good solution in a very few number of steps or a very good solution using more iterations.
which makes the \CX algorithm a very promising candidate for both, real-world quantum hardware and real-world placement problems. 

% are interesting metrics for evaluating the performance of a \place algorithm.
%Lastly, %we consider mainly the quantum annealing framework.
%It would also be interesting to find a solution using quantum gate computing.
%one could 

%\section{Acknowledgments} \label{sec:acknowledgements}
\clearpage
\bibliographystyle{ACM-Reference-Format}
\bibliography{lit}

\end{document}